\documentclass{article}
\usepackage[utf8]{inputenc}
\usepackage{fullpage}
\usepackage{amsmath, amsthm, amssymb, amsfonts, physics, enumitem, graphicx, cases, ulem, hyperref}
\usepackage{bbm, tikz, tikz-cd, amsrefs, pifont}
\hypersetup{
    colorlinks=true,
    linkcolor=black,
    urlcolor=blue,
    }
\usepackage[font=small,labelfont=bf]{caption}

\theoremstyle{definition}
\newtheorem{definition}{Definition}
\newtheorem{theorem}{Theorem}
\newtheorem{lemma}{Lemma}

\newtheorem{proposition}{Proposition}
\numberwithin{definition}{section}
\numberwithin{theorem}{section}
\numberwithin{lemma}{section}
\numberwithin{corollary}{section}
\numberwithin{proposition}{section}

\title{\Large \bf Bell pair extraction using graph foliage techniques}
\author{\normalsize Derek Zhang\footnote{Email: \href{mailto:derekzhang800@yahoo.com}{derekzhang800@yahoo.com}}  \\
\small \it Department of Physics, Amherst College, Amherst, Massachusetts 01002, USA}
\date{\small November 25, 2023}

\begin{document}

\maketitle

\begin{abstract}
    Future quantum networks can facilitate communication of quantum information between various nodes. We are particularly interested in whether multiple pairs can communicate simultaneously across a network. Quantum networks can be represented with graph states, and producing communication links amounts to performing certain quantum operations on graph states. This problem can be formulated in a graph-theoretic sense with the (Bell) vertex-minor problem. We discuss the recently introduced foliage partition and provide a generalization. This generalization leads us to a useful result for approaching the vertex-minor problem. We apply this result to identify the exact solution for the Bell vertex-minor problem on line, tree, and ring graphs.
\end{abstract}

\tableofcontents

\pagebreak

\section{Introduction}

Entanglement is simultaneously the most confounding and most useful phenomenon in quantum mechanics. Two particles in an entangled quantum state, known as a Bell pair, can be used to communicate quantum information \cite{tele}. In the realm of quantum information theory, entanglement can be exploited to perform communication over quantum networks \cite{transform}. Communication has been experimentally realized over long-distance, global networks \cite{satellite}, towards the eventual construction of a quantum internet \cite{qnet}. On a quantum network, we would like to create \textit{resource states} which can then be used for communication or other computational tasks \cite{algo}. Efficient processing of resource states is necessary due to the rapid deleterious effect of noise on quantum information \cite{noisy}. In general, generic resource states are first prepared, and then simultaneous communication requests between various pairs are processed \cite{qnrlc}. Overlapping requests may create \textit{bottlenecks}, and efficiently identifying these bottlenecks is central to problems in quantum network routing \cites{qnrlc,limits}.

Quantum networks can be represented efficiently using graph states, which are a subset of stabilizer states \cites{gstates,stabilizer}. We are interested in determining whether a specific target state can be obtained from a graph state by local Clifford operations, Pauli measurements, and classical communication. This problem was introduced as the \textit{qubit-minor problem} by ref. \cite{transform}, which showed its NP-completeness and demonstrated its equivalence to the \textit{vertex-minor problem}, which has been previously studied in graph theory contexts \cites{lc,vminor}. The question of whether multiple Bell pairs can be extracted for communication is the Bell vertex-minor problem, which is also known to be NP-complete \cite{bellvm}.

The \textit{foliage} of a graph is a useful graph-theoretic tool for addressing the vertex-minor problem. The foliage structure contains information about relationships between vertices which are preserved under local Clifford operations. This structure was first recognized by ref. \cite{transform} and named by ref. \cite{limits}. It was further elucidated by ref. \cite{fpartition}, which introduced the foliage partition and foliage graph, which can be efficiently computed. The primary goal of this paper is to apply the foliage partition and graph to the Bell vertex-minor problem.

This paper is organized follows. In Section 2, we introduce basic concepts in graph theory and provide an overview of graph states and the vertex-minor problem. In Section 3, we redefine the foliage partition and graph in a broader sense than in ref. \cite{fpartition}. In Section 4, we use the foliage to prove our main theorem, a powerful way to reduce vertex-minor problems. In Section 5, we use this theorem to completely determine the Bell vertex-minor problem for line, tree, and ring graphs. Our work expands on ref. \cite{limits}, which proved bottleneck results for line and ring graphs.

\section{Preliminaries}\label{sec:prelim}

\subsection{Graph Theory}

\begin{definition}
    A finite, undirected \textbf{graph} is a pair $G=(V,E)$ with vertices $V = \{1, \dots, N \}$ and edges $E \subset \{ \{ a, b\} \mid a,b \in V, a \neq b \} = [V]^2$.
\end{definition}

% NOTE: FIX ISOLATED
\begin{itemize}
    \item The \textbf{neighborhood} of a vertex $a$ is $N_a = \{ b \in V \mid \{a,b\} \in E\}$. The \textbf{degree} of $a$ is $\deg{a} = \abs{N_a}$. If $\deg{a}=0$, $a$ is an \textbf{isolated vertex}, and if $\deg{a} = 1$, $a$ is a \textbf{leaf}.
    \item Two graphs $G_1=(V_1,E_1)$, $G_2=(V_2,E_2)$ are \textbf{isomorphic} if there exists a bijection $f:V_1 \to V_2$ such that $\{a,b\} \in E_1 \iff \{f(a),f(b)\} \in E_2$. Graph theory is mainly concerned with properties of graphs that are invariant under isomorphism, or a permutation of vertices.
    \item For a graph $G$ with vertex $a$, $G\setminus a$ (\textbf{vertex deletion}) is the graph obtained by deleting $a$ and all edges connected to it. For a set of vertices $U$, $G\setminus U$ is defined similarly.
    \item For a graph $G = (V,E)$ and a set of edges $F$, $G \cup F = (V, E\cup F)$ and $G + F = (V, E+F)$, where $E+F$ corresponds to addition modulo 2 or the symmetric difference $(E \cup F) \setminus (E \cap F)$.
    \item The \textbf{complement} of a graph $G=(V,E)$ is $G'=(V,[V]^2\setminus E)$.
    \item The \textbf{subgraph} $G[A]$ induced by a subset of vertices $A$ on $G$ is the graph with vertices $A$ and all edges between them that were in $G$,
\end{itemize}

\subsection{Graph States and LC-Equivalence}

For any graph, we consider a system of qubits labeled by the vertices, corresponding to the Hilbert space $\mathcal{H}^V = \bigotimes_{a\in V} \mathcal{H}^a \cong \bigotimes_{a\in V} \mathbb{C}^2$. Edges can be thought of abstractly as interactions between pairs of qubits $\{a,b\}$. We represent this interaction with the Controlled-Z gate on qubits $a$ and $b$, $CZ_{ab} = \mathbbm{1}_{ab} - 2\ket{11}^{ab}\bra{11}$.

\begin{definition}
    Given a graph $G(V,E)$, define its \textbf{graph state} to be
    \begin{align}
        \ket{G} = \prod_{\{a,b\} \in E}CZ_{ab} \ket{+}^V,
    \end{align}\label{def:graphstate1}
\end{definition}

\vspace{-10pt}
\noindent where $\ket{+}^V = \bigotimes_{a\in V} \ket{+}^a$ and $\ket{+} = \frac{1}{\sqrt{2}}(\ket{0}+\ket{1})$. To prepare qubits in a graph state, initialize every qubit in the state $\ket{+}$ and then apply $CZ$ gates for each edge.

\begin{figure}[h]
\begin{center}
\begin{tikzpicture}[scale=0.75, transform shape]
  \node[circle, fill=purple!60] (1) at (0,0) {\textcolor{white}{1}};
  \node[circle, fill=purple!60] (2) at (2,0) {\textcolor{white}{2}};
  \node[circle, fill=purple!60] (3) at (2,2) {\textcolor{white}{3}};
  \node[circle, fill=purple!60] (4) at (0,2) {\textcolor{white}{4}};
  
  \draw[gray, line width=1.5] (1) -- (2);
  \draw[gray, line width=1.5] (2) -- (3);
  \draw[gray, line width=1.5] (2) -- (4);

  \node at (1,-0.3) {$CZ_{12}$};
  \node at (1.2,1.4) {$CZ_{24}$};
  \node at (2.5,1) {$CZ_{23}$};
\end{tikzpicture}
\end{center}
\caption{A four vertex graph with corresponding $CZ$ gates.}
\label{fig:4v}
\end{figure}
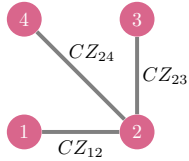

For an in-depth overview of graph states, see ref. \cite{gstates}. While every graph corresponds to a unique graph state, we can define a notion of equivalence between different graph states using local Clifford operations.

\begin{definition}
    Two graphs $G=(V,E)$ and $G'=(V',E')$ are \textbf{LC-equivalent}, denoted $G \sim_{\text{LC}} G'$, if there exists a local Clifford unitary $U \in \mathcal{C}_1^V$ such that $\ket{G'} = U\ket{G}$.
\end{definition}

The local Clifford group is $\mathcal{C}_1^V = \{U \in U(2)^V \mid U\mathcal{P}^V U^* = \mathcal{P}^V\}$, where $\mathcal{P}$ is the Pauli group generated by the Pauli matrices. It can be shown that elements in $\mathcal{C}_1^V$ are composed of sequences of $H$ and $S$ gates \cite{gstates}. The group properties of local Clifford unitaries give that LC-equivalence is an equivalence relation. LC-equivalence can be conveniently formulated in terms of an operation on graphs.

\begin{definition}
    The \textbf{local complement} of a graph $G$ at vertex $a$ is obtained by complementing the subgraph $G[N_a]$, i.e. $\tau_a(G) = G + K_{N_a}$, where $K_{N_a}$ is the complete graph on the vertices of $N_a$.
\end{definition}

\begin{figure}[h]
\begin{center}
\begin{tikzpicture}[scale=0.75, transform shape]
  \node[circle, fill=purple!60] (1) at (0,0) {\textcolor{white}{1}};
  \node[circle, fill=purple!80] (2) at (2,0) {\textcolor{white}{2}};
  \node[circle, fill=purple!60] (3) at (2,2) {\textcolor{white}{3}};
  \node[circle, fill=purple!60] (4) at (0,2) {\textcolor{white}{4}};
  
  \draw[gray, line width=1.5] (1) -- (2);
  \draw[gray, line width=1.5] (2) -- (3);
  \draw[gray, line width=1.5] (2) -- (4);
  \draw[gray, line width=1.5] (1) -- (3);
  
  \draw[<->] (3,1) -- (4,1);
  \node at (3.5,1.25) {$\tau_2$};

  \node[circle, fill=purple!60] (1) at (5,0) {\textcolor{white}{1}};
  \node[circle, fill=purple!80] (2) at (7,0) {\textcolor{white}{2}};
  \node[circle, fill=purple!60] (3) at (7,2) {\textcolor{white}{3}};
  \node[circle, fill=purple!60] (4) at (5,2) {\textcolor{white}{4}};
  
  \draw[gray, line width=1.5] (1) -- (2);
  \draw[gray, line width=1.5] (2) -- (3);
  \draw[gray, line width=1.5] (2) -- (4);
  \draw[gray, line width=1.5] (1) -- (4);
  \draw[gray, line width=1.5] (3) -- (4);
\end{tikzpicture}
\end{center}
\caption{Local complementation on vertex 2 of a four vertex graph. Observe that $\tau_a = \tau_a^{-1}$.}
\label{fig:lc}
\end{figure}
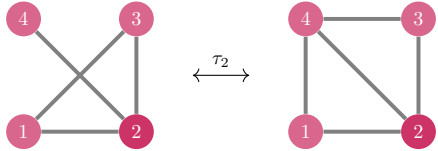

\begin{theorem}[Proposition 5 in \cite{gstates}]
    Given a graph $G$, its local complement at a vertex $a$ yields an LC-equivalent graph  $\tau_a(G)$. Furthermore, any two LC-equivalent graphs are related by a sequence of local complementations.\footnote{Conveniently, ``LC" corresponds to both local Clifford and local complementation; we will use it to abbreviate the latter.} \label{theo:sequence}
\end{theorem}

The \textit{local} operations involved in LC-equivalence preserve entanglement \textit{between} vertices in some sense. As such, studying graphs up to LC-equivalence is a proxy for studying graphs up to their entanglement properties. It is worth noting here that since our vertices represent physical qubits, they must be labeled. As such, isomorphic graphs are frequently not LC-equivalent.

\subsection{Vertex-Minor Problem}

In quantum mechanics, it is well known that measurement of a particle ``collapses" its wave function. In graph states, this translates conveniently to vertex deletion. The graph operations $X_a,Y_a,Z_a$ correspond to Pauli (projective) measurements on qubits $P_{i,\pm}^a = \frac{1 \pm \sigma_i^a}{2}$ in the $i=x,y,z$ bases: \cite{gstates}
\begin{align}
    P_{z,\pm}^a \ket{G} &= \ket{G \setminus a} = \ket{Z_a(G)}\\
    P_{y,\pm}^a \ket{G} &= \ket{\tau_a(G) \setminus a} = \ket{Y_a(G)}\\
    P_{x,\pm}^a \ket{G} &= \ket{\tau_b\circ\tau_a\circ\tau_b(G) \setminus a} = \ket{X_a(G)} \label{eq:paulix}
\end{align}

For $P_{x,\pm}^a$, we choose some $b \in N_a$ and write $X_a^{(b)}$; all choices will yield LC-equivalent graphs.\footnote{For the special case $N_a = \emptyset$, take $X_a = Z_a$.}

\begin{figure}[h]
\begin{center}
\begin{tikzpicture}[scale=0.75, transform shape]
  \node[circle, fill=purple!60] (1) at (0,0) {\textcolor{white}{1}};
  \node[circle, fill=purple!80] (2) at (2,0) {\textcolor{white}{2}};
  \node[circle, fill=purple!60] (3) at (2,2) {\textcolor{white}{3}};
  \node[circle, fill=purple!60] (4) at (0,2) {\textcolor{white}{4}};
  
  \draw[gray, line width=1.5] (1) -- (2);
  \draw[gray, line width=1.5] (2) -- (3);
  \draw[gray, line width=1.5] (2) -- (4);
  \draw[gray, line width=1.5] (1) -- (3);
  
  \draw[->] (1,-0.5) -- (1,-1.5);
  \node at (1.3,-1) {$Y_2$};

  \node[circle, fill=purple!60] (1) at (0,-4) {\textcolor{white}{1}};
  \node[circle, fill=purple!60] (3) at (2,-2) {\textcolor{white}{3}};
  \node[circle, fill=purple!60] (4) at (0,-2) {\textcolor{white}{4}};
  
  \draw[gray, line width=1.5] (1) -- (4);
  \draw[gray, line width=1.5] (3) -- (4);

  \draw[->] (-1.5,-0.5) -- (-2.5,-1.5);
  \node at (-1.5,-1) {$Z_2$};
  
  \node[circle, fill=purple!60] (1) at (-4,-4) {\textcolor{white}{1}};
  \node[circle, fill=purple!60] (3) at (-2,-2) {\textcolor{white}{3}};
  \node[circle, fill=purple!60] (4) at (-4,-2) {\textcolor{white}{4}};
  
  \draw[gray, line width=1.5] (1) -- (3);

  \draw[->] (3.5,-0.5) -- (4.5,-1.5);
  \node at (4.5,-1) {$X_2$};
  
  \node[circle, fill=purple!60] (1) at (4,-4) {\textcolor{white}{1}};
  \node[circle, fill=purple!60] (3) at (6,-2) {\textcolor{white}{3}};
  \node[circle, fill=purple!60] (4) at (4,-2) {\textcolor{white}{4}};
  
  \draw[gray, line width=1.5] (1) -- (3);
  \draw[gray, line width=1.5] (1) -- (4);
  \draw[gray, line width=1.5] (3) -- (4);
  
\end{tikzpicture}
\end{center}
\caption{Pauli measurements on vertex 2 of a four vertex graph. Coincidentally, all three choices of neighbors for the $x$-basis measurement yield the same graph.}
\label{fig:paulim}
\end{figure}
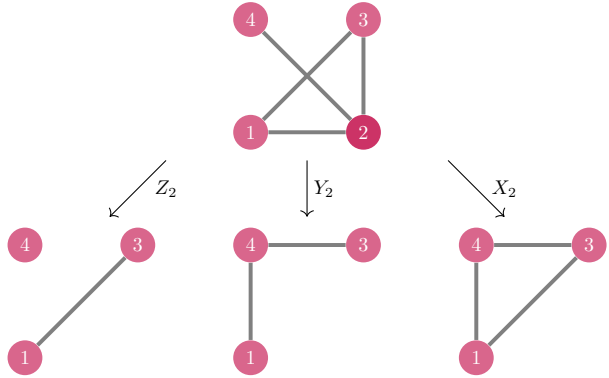

\begin{definition}
    A graph $H$ is a \textbf{vertex-minor} of a graph $G$, denoted $H<G$, if it can be obtained from $G$ through local complementations and vertex deletions.
\end{definition}

We refer to $G$ as the ``source graph" and $H$ as the ``target graph." The decision problem of whether $H<G$ is the \textbf{vertex-minor problem} \cite{vminor}, which in general is NP-complete \cite{transform}. It was shown in ref. \cite{transform} that the vertex-minor problem is equivalent to the qubit-minor problem \textemdash\ whether one graph state can be obtained from another through local Clifford operations, Pauli measurements, and classical communication.

\begin{proposition}
    Given graphs $G_1<G_2$ and $G_2<G_3$, we have $G_1<G_3$.\label{prop:vminortransitive}
\end{proposition}
\begin{proof}
    The two sequences of LCs and vertex deletions associated with $G_1<G_2$ and $G_2<G_3$ can be concatenated into another sequence for $G_1<G_3$.
\end{proof}

While the vertex-minor problem can be formulated without Pauli measurements, they are still useful, as exhibited in the following result.

\begin{theorem}[Theorem 3.2 in \cite{transform}]
    Given graphs $G(V,E)$ and $G'(V',E')$ with $V' \subset V$, let $V \setminus V' = \{ v_1, \dots v_k\}$. Then $G'<G$ iff $G' \sim_{\text{LC}} (P_{v_1} \circ \dots \circ P_{v_k}) (G)$ for some choices of $P_{v_i} \in \{X_{v_i},Y_{v_i},Z_{v_i}\}$.\label{theo:preduction}
\end{theorem}

\section{Foliage}

The foliage structure of a graph is a way to associate vertices which retain a \textit{foliage-equivalence} under LCs. This structure provides us with a means to simplify the vertex-minor problem for any arbitrary graph.

\begin{definition}
    Given a graph $G(V,E)$,
    \begin{enumerate}[label=(\alph*)]
        \item Vertex $v$ is a \textbf{leaf} if it has degree 1 ($\abs{N_v}=1$). Its unique neighbor is an \textbf{axil}. Define $L_G = \{(l,a) \mid l,a \text{ are a leaf-axil pair} \}$.
        \item Vertex $v$ is a \textbf{twin} if $\exists w\neq v$ such that $N_v \setminus w = N_w \setminus v \neq \emptyset$. Define $T_G = \{\{t_1,t_2\} \mid t_1,t_2 \text{ are a twin pair} \}$.
        \item The \textbf{foliage} is the set of all leaves, axils, and twins of $G$.
    \end{enumerate}
    
\end{definition}

The foliage was shown to be LC-invariant (unchanged under LCs) in ref. \cite{limits}.

\subsection{Foliage Partition}

\begin{definition} [Definition 2 in \cite{fpartition}]
    Given a graph $G$, two vertices $v,w\in G$ are \textbf{foliage-equivalent}, denoted $v \sim_\text{F} w$, if $v=w$, $(v,w)\in L_G$, $(w,v)\in L_G$, or $\{v,w\} \in T_G$.
\end{definition}

\begin{definition}
    A \textbf{partition} of a set of vertices $V$ is a set of disjoint subsets $A = \{A_1,\dots,A_k\}$ with $A_i \subset V$ and $\bigcup_i A_i = V$. Given two partitions $A$ and $B$ of $V$ with every $A_i$ contained in some $B_j$, we write $A\leq B$. We say $A$ is a \textbf{finer partition} than $B$, and $B$ is a \textbf{coarser partition} than $A$.
\end{definition}

For any partition $A$, we have $A \leq A$. The finest partition contains each vertex as a singleton, and the coarsest partition is $\{V\}$.

\begin{proposition}
    Foliage-equivalence is an equivalence relation.\label{prop:feer}
\end{proposition}

It is well known that the set of equivalence classes is a partition. From this fact we can define the canonical foliage partition.

\begin{definition}
    Given a graph $G$, its \textbf{canonical foliage partition} $F(V) = \{ V_1, \dots, V_k \}$ is the set of foliage-equivalence classes.\footnote{Ref. \cite{fpartition} used the notation $\hat{V}$ for the canonical foliage partition and $\hat{G}$ for the canonical foliage graph. We introduce $F(V)$ and $F(G)$, respectively, to accommodate a generalization.}
\end{definition}

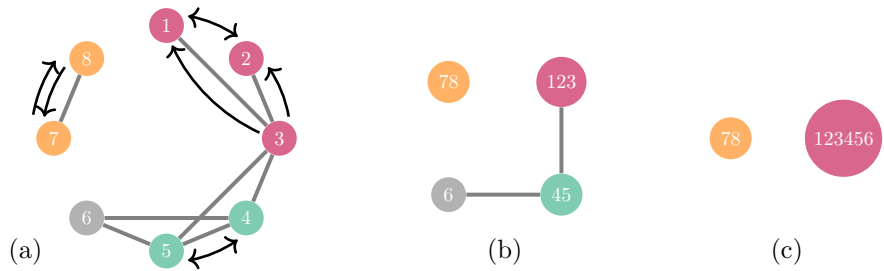
\begin{figure}[h]
\begin{center}
\begin{tikzpicture}[scale=0.75, transform shape]
  \node[circle, fill=purple!60] (1) at (0,2) {\textcolor{white}{1}};
  \node[circle, fill=purple!60] (2) at (1.414,1.414) {\textcolor{white}{2}};
  \node[circle, fill=purple!60] (3) at (2,0) {\textcolor{white}{3}};
  \node[circle, fill=green!60!blue!50] (4) at (1.414,-1.414) {\textcolor{white}{4}};
  \node[circle, fill=green!60!blue!50] (5) at (0,-2) {\textcolor{white}{5}};
  \node[circle, fill=gray!60] (6) at (-1.414,-1.414) {\textcolor{white}{6}};
  \node[circle, fill=orange!60] (7) at (-2,0) {\textcolor{white}{7}};
  \node[circle, fill=orange!60] (8) at (-1.414,1.414) {\textcolor{white}{8}};
  
  \draw[gray, line width=1.5] (1) -- (3);
  \draw[gray, line width=1.5] (2) -- (3);
  \draw[gray, line width=1.5] (3) -- (4);
  \draw[gray, line width=1.5] (3) -- (5);
  \draw[gray, line width=1.5] (4) -- (5);
  \draw[gray, line width=1.5] (4) -- (6);
  \draw[gray, line width=1.5] (5) -- (6);
  \draw[gray, line width=1.5] (7) -- (8);

  \draw[line width=1, <->] (0,0)+(55:2.2cm) arc (55:80:2.2cm);
  \draw[line width=1, <->] (0,0)+(280:2.2cm) arc (280:305:2.2cm);
  \draw[line width=1, ->] (0,0)+(10:2.2cm) arc (10:35:2.2cm);
  \draw[line width=1, <-] (3,3)+(205:3.2cm) arc (205:245:3.2cm);
  \draw[line width=1, ->] (0,0)+(145:2.2cm) arc (145:170:2.2cm);
  \draw[line width=1, <-] (0,0)+(145:2.4cm) arc (145:170:2.4cm);
  %\draw[line width=1, <->] (-3.4,1.4)+(325:2.2cm) arc (325:350:2.2cm);

  %Foliage graph
  \node[circle, fill=gray!60] (1) at (5,-1) {\textcolor{white}{6}};
  \node[circle, fill=green!60!blue!50] (2) at (7,-1) {\textcolor{white}{45}};
  \node[circle, fill=purple!60] (3) at (7,1) {\textcolor{white}{123}};
  \node[circle, fill=orange!60] (4) at (5,1) {\textcolor{white}{78}};
  
  \draw[gray, line width=1.5] (1) -- (2);
  \draw[gray, line width=1.5] (2) -- (3);

  %2nd-foliage graph
  \node[circle, fill=purple!60] (3) at (12,0) {\textcolor{white}{123456}};
  \node[circle, fill=orange!60] (4) at (10,0) {\textcolor{white}{78}};

  \begin{scope}[scale=1.333,transform shape]
      \node at (-1.875,-1.5) {(a)};
      \node at (4.5,-1.5) {(b)};
      \node at (8.25,-1.5) {(c)};
  \end{scope}

\end{tikzpicture}
\end{center}
\caption{(a) An eight vertex graph with canonical foliage partition indicated by colors. The foliage is composed of all vertices except vertex 5. Twin relationships are indicated by $t_1\longleftrightarrow t_2$ and leaf-axil relationships by $l\longleftarrow a$. Notice the multiple overlapping relationships. (b) The corresponding canonical foliage graph $F(G)$. (c) The 2nd-foliage graph $F^2(G)$. For convenience, we write `123' in place of $\{1,2,3\}$ and `123456' in place of $\{\{1,2,3\},\{4,5\},\{6\} \}$.}
\label{fig:fpartition}
\end{figure}

\begin{theorem}[Proposition 7 in \cite{fpartition}]
    The canonical foliage partition is LC-invariant.\label{theo:fplci}
\end{theorem}

Proposition \ref{prop:feer} and Theorem \ref{theo:fplci} are proven in the Appendix (ref. \cite{fpartition} proves these in greater generality).

The essence of Theorem \ref{theo:fplci} is that LCs transform twin, leaf-axil, and axil-leaf pairs between each other. Ref. \cite{fpartition} showed that every foliage-equivalence class is either a singleton, a star graph (an axil and its leaves), a totally connected set (clique) of twins, or a totally disconnected set (anticlique) of twins, and LCs transform between these types.

\begin{definition}
    Given a graph $G$ with some partition $F_W(V) = \{W_1,\dots,W_l\}$ of $V$, $F_W(V)$ is a \textbf{foliage partition} on $W$ if $F_W(V) \leq F(V) := F_V(V)$.
\end{definition}

Here, vertices within a partition subset are still equivalent to each other, but can also be equivalent to outside vertices. A foliage partition is ``LC-invariant" in the sense that vertices in partition subsets remain equivalent under LCs.

\begin{figure}[htb]
\centering
\begin{tikzpicture}[scale=0.64, transform shape]
    %axil, 3
    \node[circle, fill=purple!60] (1) at (-8, -1) {\textcolor{white}{1}};
    \node[circle, fill=purple!60] (2) at (-6.586, -1.586) {\textcolor{white}{2}};
    \node[circle, fill=purple!60] (3) at (-6, -3) {\textcolor{white}{3}};
    \node[circle, fill=green!60!blue!25] (4) at (-6.586, -4.414) {\textcolor{white}{4}};
    \node[circle, fill=green!60!blue!25] (5) at (-8, -5) {\textcolor{white}{5}};
    \node[circle, fill=gray!30] (6) at (-9.414, -4.414) {\textcolor{white}{6}};
    \node[circle, fill=orange!30] (7) at (-10, -3) {\textcolor{white}{7}};
    \node[circle, fill=orange!30] (8) at (-9.414, -1.586) {\textcolor{white}{8}};
    
    \draw[gray, line width=1.5] (1) -- (3);
    \draw[gray, line width=1.5] (2) -- (3);
    \draw[gray!50, line width=1.5] (3) -- (7);
    \draw[gray!50, line width=1.5] (3) -- (8);
    \draw[gray!50, line width=1.5] (4) -- (5);
    \draw[gray!50, line width=1.5] (4) -- (6);
    \draw[gray!50, line width=1.5] (5) -- (6);
    \draw[gray!50, line width=1.5] (4) -- (7);
    \draw[gray!50, line width=1.5] (4) -- (8);
    \draw[gray!50, line width=1.5] (5) -- (7);
    \draw[gray!50, line width=1.5] (5) -- (8);

    %axil,2
    \node[circle, fill=purple!60] (1) at (-1, -1) {\textcolor{white}{1}};
    \node[circle, fill=purple!60] (2) at (0.414, -1.586) {\textcolor{white}{2}};
    \node[circle, fill=purple!60] (3) at (1, -3) {\textcolor{white}{3}};
    \node[circle, fill=green!60!blue!25] (4) at (0.414, -4.414) {\textcolor{white}{4}};
    \node[circle, fill=green!60!blue!25] (5) at (-1, -5) {\textcolor{white}{5}};
    \node[circle, fill=gray!30] (6) at (-2.414, -4.414) {\textcolor{white}{6}};
    \node[circle, fill=orange!30] (7) at (-3, -3) {\textcolor{white}{7}};
    \node[circle, fill=orange!30] (8) at (-2.414, -1.586) {\textcolor{white}{8}};
    
    \draw[gray, line width=1.5] (1) -- (2);
    \draw[gray, line width=1.5] (3) -- (2);
    \draw[gray!50, line width=1.5] (2) -- (7);
    \draw[gray!50, line width=1.5] (2) -- (8);
    \draw[gray!50, line width=1.5] (4) -- (5);
    \draw[gray!50, line width=1.5] (4) -- (6);
    \draw[gray!50, line width=1.5] (5) -- (6);
    \draw[gray!50, line width=1.5] (4) -- (7);
    \draw[gray!50, line width=1.5] (4) -- (8);
    \draw[gray!50, line width=1.5] (5) -- (7);
    \draw[gray!50, line width=1.5] (5) -- (8);

    %twins, con
    \node[circle, fill=purple!60] (1) at (6, -1) {\textcolor{white}{1}};
    \node[circle, fill=purple!60] (2) at (7.414, -1.586) {\textcolor{white}{2}};
    \node[circle, fill=purple!60] (3) at (8, -3) {\textcolor{white}{3}};
    \node[circle, fill=green!60!blue!25] (4) at (7.414, -4.414) {\textcolor{white}{4}};
    \node[circle, fill=green!60!blue!25] (5) at (6, -5) {\textcolor{white}{5}};
    \node[circle, fill=gray!30] (6) at (4.586, -4.414) {\textcolor{white}{6}};
    \node[circle, fill=orange!30] (7) at (4, -3) {\textcolor{white}{7}};
    \node[circle, fill=orange!30] (8) at (4.586, -1.586) {\textcolor{white}{8}};
    
    \draw[gray!50, line width=1.5] (1) -- (7);
    \draw[gray!50, line width=1.5] (1) -- (8);
    \draw[gray!50, line width=1.5] (2) -- (7);
    \draw[gray!50, line width=1.5] (2) -- (8);
    \draw[gray!50, line width=1.5] (3) -- (7);
    \draw[gray!50, line width=1.5] (3) -- (8);
    \draw[gray!50, line width=1.5] (4) -- (5);
    \draw[gray!50, line width=1.5] (4) -- (6);
    \draw[gray!50, line width=1.5] (5) -- (6);
    \draw[gray!50, line width=1.5] (4) -- (7);
    \draw[gray!50, line width=1.5] (4) -- (8);
    \draw[gray!50, line width=1.5] (5) -- (7);
    \draw[gray!50, line width=1.5] (5) -- (8);
    \draw[gray!50, line width=1.5] (7) -- (8);
    \draw[gray, line width=1.5] (1) -- (2);
    \draw[gray, line width=1.5] (1) -- (3);
    \draw[gray, line width=1.5] (2) -- (3);

    %twins, discon
    \node[circle, fill=purple!60] (1) at (13, -1) {\textcolor{white}{1}};
    \node[circle, fill=purple!60] (2) at (14.414, -1.586) {\textcolor{white}{2}};
    \node[circle, fill=purple!60] (3) at (15, -3) {\textcolor{white}{3}};
    \node[circle, fill=green!60!blue!25] (4) at (14.414, -4.414) {\textcolor{white}{4}};
    \node[circle, fill=green!60!blue!25] (5) at (13, -5) {\textcolor{white}{5}};
    \node[circle, fill=gray!30] (6) at (11.586, -4.414) {\textcolor{white}{6}};
    \node[circle, fill=orange!30] (7) at (11, -3) {\textcolor{white}{7}};
    \node[circle, fill=orange!30] (8) at (11.586, -1.586) {\textcolor{white}{8}};
    
    \draw[gray!50, line width=1.5] (1) -- (8);
    \draw[gray!50, line width=1.5] (2) -- (8);
    \draw[gray!50, line width=1.5] (3) -- (8);
    \draw[gray!50, line width=1.5] (1) -- (4);
    \draw[gray!50, line width=1.5] (2) -- (4);
    \draw[gray!50, line width=1.5] (3) -- (4);
    \draw[gray!50, line width=1.5] (1) -- (5);
    \draw[gray!50, line width=1.5] (2) -- (5);
    \draw[gray!50, line width=1.5] (3) -- (5);
    \draw[gray!50, line width=1.5] (4) -- (6);
    \draw[gray!50, line width=1.5] (5) -- (6);
    \draw[gray!50, line width=1.5] (4) -- (8);
    \draw[gray!50, line width=1.5] (5) -- (8);
    \draw[gray!50, line width=1.5] (7) -- (8);
    
    %arrow
    \draw[line width=0.5, <->] (-1,-12)+(112:12cm) arc (112:68:12cm);
    \node at (-1,0.3) {$\tau_3$};
    \draw[<->] (2,-3)--(3,-3);
    \node at (2.5,-2.7) {$\tau_2$};
    \draw[<->] (9,-3)--(10,-3);
    \node at (9.5,-2.7) {$\tau_8$};

    %generalized
    %axil, 3
    \node[circle, fill=blue!15] (1) at (-8, -8) {\textcolor{white}{1}};
    \node[circle, fill=purple!60] (2) at (-6.586, -8.586) {\textcolor{white}{2}};
    \node[circle, fill=purple!60] (3) at (-6, -10) {\textcolor{white}{3}};
    \node[circle, fill=green!60!blue!25] (4) at (-6.586, -11.414) {\textcolor{white}{4}};
    \node[circle, fill=blue!60!purple!15] (5) at (-8, -12) {\textcolor{white}{5}};
    \node[circle, fill=gray!30] (6) at (-9.414, -11.414) {\textcolor{white}{6}};
    \node[circle, fill=orange!30] (7) at (-10, -10) {\textcolor{white}{7}};
    \node[circle, fill=orange!30] (8) at (-9.414, -8.586) {\textcolor{white}{8}};

    \draw[gray!50, line width=1.5] (1) -- (3);
    \draw[gray, line width=1.5] (2) -- (3);
    \draw[gray!50, line width=1.5] (3) -- (7);
    \draw[gray!50, line width=1.5] (3) -- (8);
    \draw[gray!50, line width=1.5] (4) -- (5);
    \draw[gray!50, line width=1.5] (4) -- (6);
    \draw[gray!50, line width=1.5] (5) -- (6);
    \draw[gray!50, line width=1.5] (4) -- (7);
    \draw[gray!50, line width=1.5] (4) -- (8);
    \draw[gray!50, line width=1.5] (5) -- (7);
    \draw[gray!50, line width=1.5] (5) -- (8);

    %axil,2
    \node[circle, fill=blue!15] (1) at (-1, -8) {\textcolor{white}{1}};
    \node[circle, fill=purple!60] (2) at (0.414, -8.586) {\textcolor{white}{2}};
    \node[circle, fill=purple!60] (3) at (1, -10) {\textcolor{white}{3}};
    \node[circle, fill=green!60!blue!25] (4) at (0.414, -11.414) {\textcolor{white}{4}};
    \node[circle, fill=blue!60!purple!15] (5) at (-1, -12) {\textcolor{white}{5}};
    \node[circle, fill=gray!30] (6) at (-2.414, -11.414) {\textcolor{white}{6}};
    \node[circle, fill=orange!30] (7) at (-3, -10) {\textcolor{white}{7}};
    \node[circle, fill=orange!30] (8) at (-2.414, -8.586) {\textcolor{white}{8}};

    \draw[gray!50, line width=1.5] (1) -- (2);
    \draw[gray, line width=1.5] (3) -- (2);
    \draw[gray!50, line width=1.5] (2) -- (7);
    \draw[gray!50, line width=1.5] (2) -- (8);
    \draw[gray!50, line width=1.5] (4) -- (5);
    \draw[gray!50, line width=1.5] (4) -- (6);
    \draw[gray!50, line width=1.5] (5) -- (6);
    \draw[gray!50, line width=1.5] (4) -- (7);
    \draw[gray!50, line width=1.5] (4) -- (8);
    \draw[gray!50, line width=1.5] (5) -- (7);
    \draw[gray!50, line width=1.5] (5) -- (8);
    
    %twins, con
    \node[circle, fill=blue!15] (1) at (6, -8) {\textcolor{white}{1}};
    \node[circle, fill=purple!60] (2) at (7.414, -8.586) {\textcolor{white}{2}};
    \node[circle, fill=purple!60] (3) at (8, -10) {\textcolor{white}{3}};
    \node[circle, fill=green!60!blue!25] (4) at (7.414, -11.414) {\textcolor{white}{4}};
    \node[circle, fill=blue!60!purple!15] (5) at (6, -12) {\textcolor{white}{5}};
    \node[circle, fill=gray!30] (6) at (4.586, -11.414) {\textcolor{white}{6}};
    \node[circle, fill=orange!30] (7) at (4, -10) {\textcolor{white}{7}};
    \node[circle, fill=orange!30] (8) at (4.586, -8.586) {\textcolor{white}{8}};

    \draw[gray!50, line width=1.5] (1) -- (2);
    \draw[gray!50, line width=1.5] (1) -- (3);
    \draw[gray, line width=1.5] (2) -- (3);
    \draw[gray!50, line width=1.5] (1) -- (7);
    \draw[gray!50, line width=1.5] (1) -- (8);
    \draw[gray!50, line width=1.5] (2) -- (7);
    \draw[gray!50, line width=1.5] (2) -- (8);
    \draw[gray!50, line width=1.5] (3) -- (7);
    \draw[gray!50, line width=1.5] (3) -- (8);
    \draw[gray!50, line width=1.5] (4) -- (5);
    \draw[gray!50, line width=1.5] (4) -- (6);
    \draw[gray!50, line width=1.5] (5) -- (6);
    \draw[gray!50, line width=1.5] (4) -- (7);
    \draw[gray!50, line width=1.5] (4) -- (8);
    \draw[gray!50, line width=1.5] (5) -- (7);
    \draw[gray!50, line width=1.5] (5) -- (8);
    \draw[gray!50, line width=1.5] (7) -- (8);

    %twins, discon
    \node[circle, fill=blue!15] (1) at (13, -8) {\textcolor{white}{1}};
    \node[circle, fill=purple!60] (2) at (14.414, -8.586) {\textcolor{white}{2}};
    \node[circle, fill=purple!60] (3) at (15, -10) {\textcolor{white}{3}};
    \node[circle, fill=green!60!blue!25] (4) at (14.414, -11.414) {\textcolor{white}{4}};
    \node[circle, fill=blue!60!purple!15] (5) at (13, -12) {\textcolor{white}{5}};
    \node[circle, fill=gray!30] (6) at (11.586, -11.414) {\textcolor{white}{6}};
    \node[circle, fill=orange!30] (7) at (11, -10) {\textcolor{white}{7}};
    \node[circle, fill=orange!30] (8) at (11.586, -8.586) {\textcolor{white}{8}};

    \draw[gray!50, line width=1.5] (1) -- (8);
    \draw[gray!50, line width=1.5] (2) -- (8);
    \draw[gray!50, line width=1.5] (3) -- (8);
    \draw[gray!50, line width=1.5] (1) -- (4);
    \draw[gray!50, line width=1.5] (2) -- (4);
    \draw[gray!50, line width=1.5] (3) -- (4);
    \draw[gray!50, line width=1.5] (1) -- (5);
    \draw[gray!50, line width=1.5] (2) -- (5);
    \draw[gray!50, line width=1.5] (3) -- (5);
    \draw[gray!50, line width=1.5] (4) -- (6);
    \draw[gray!50, line width=1.5] (5) -- (6);
    \draw[gray!50, line width=1.5] (4) -- (8);
    \draw[gray!50, line width=1.5] (5) -- (8);
    \draw[gray!50, line width=1.5] (7) -- (8);

    %arrow
    \draw[line width=0.5, <->] (-1,-19)+(112:12cm) arc (112:68:12cm);
    \node at (-1,-6.7) {$\tau_3$};
    \draw[<->] (2,-10)--(3,-10);
    \node at (2.5,-9.7) {$\tau_2$};
    \draw[<->] (9,-10)--(10,-10);
    \node at (9.5,-9.7) {$\tau_8$};
    
\end{tikzpicture}
\vspace{-100pt}
\caption{LC-invariance demonstrated for a partition subset of the canonical foliage partition and a finer foliage partition of the same graph. This also demonstrates that within a set of equivalent vertices, any leaf-axil or twin configuration can be obtained through LCs. Crucially, it is not necessary that the partition subset is maximal.}
\end{figure}
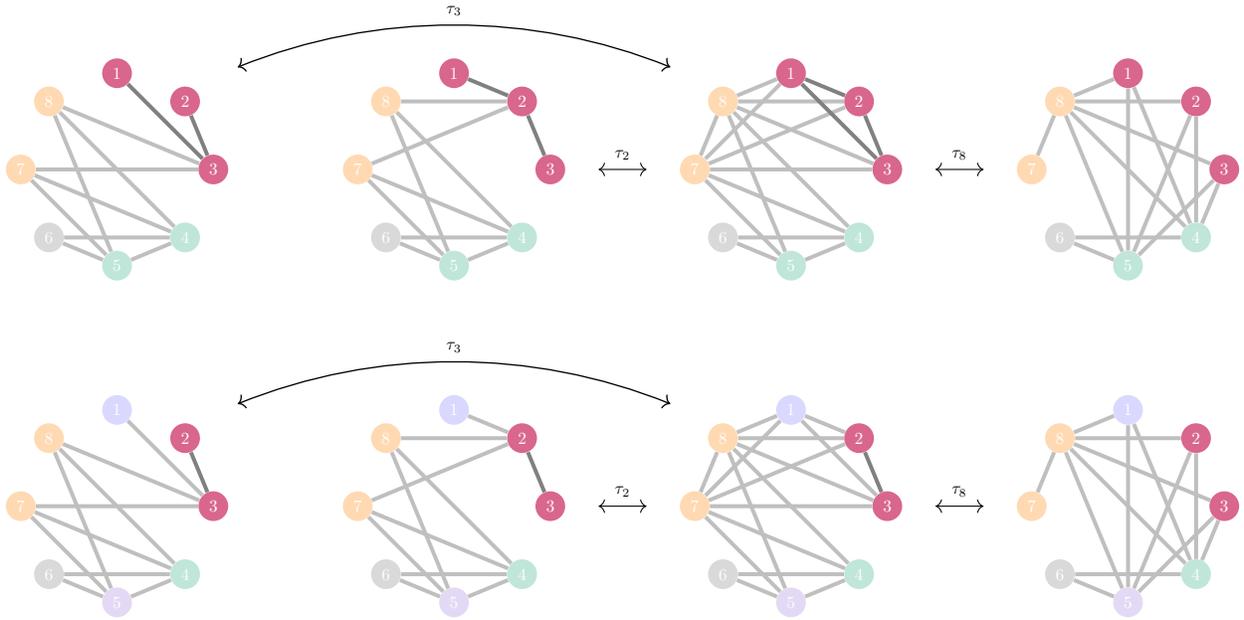

\subsection{Foliage Graph}

\begin{definition}
    Given a graph $G(V,E)$, its \textbf{canonical foliage graph} is $F(G)(F(V),F(E))$, where $\{V_i,V_j\}\in F(E) \iff \exists v_i\in V_i,v_j\in V_j$ such that $\{v_i,v_j\} \in E$.
\end{definition}

\begin{definition}
    Given a graph $G(V,E)$ with foliage partition $F_W(V)$, the \textbf{foliage graph} associated with $F_W(V)$ is $F_W(G)(F_W(V),F_W(E))$, where $\{W_i,W_j\}\in F_W(E) \iff \exists w_i\in W_i,w_j\in W_j$ such that $\{w_i,w_j\} \in E$.\footnote{Let $F(G) := F_V(G)$.}
\end{definition}

We can apply this graph operation $F_W$ indefinitely. For a graph $G$, call $F^n(G)$ the \textbf{\textit{n}th-foliage graph}. An example of the 2nd-foliage graph is given in Figure \ref{fig:fpartition}c.

\begin{theorem}[Lifted Local Complementation]
    Let graph $G$ have foliage partition $F_W(V)$. Then $F_W(\tau_a(G)) = \tau_{W_a}(F_W(G))$ for $a\in G$ with $\abs{N_a} > 1$ and corresponding $W_a\in F_W(G)$.
    \label{theo:liftedlc}
\end{theorem}
\begin{proof}
    This result follows from Observation 3 in \cite{fpartition}, where lifted local complementation was proven for $F_W=F_V=F$. The general case is identical, since we do not need non-equivalence of vertices in different partition subsets.
\end{proof}

Lifted local complementation allows us to easily perform LCs on a foliage graph. As the theorem suggests, if two graphs are LC-equivalent, their canonical foliage graphs are LC-equivalent. If we choose the same foliage partition, then the corresponding foliage graphs are also LC-equivalent.

The key observation of this section is that we do not need to take the coarsest partition with equivalent vertices. The power of this is that we can perform lifted local complementation on many choices of a foliage graph.

\section{Foliage and Vertex-Minors}

In order to use foliage graphs to simplify the vertex-minor problem, we would like to know how LCs and vertex deletions affect the foliage graph structure. LCs behave conveniently according to lifted local complementation, but vertex deletion turns out to be more troublesome.

We begin by stating an existing method for simplifying the vertex-minor problem, which we will refer to as \textit{source reduction}.

\begin{theorem}[Source Reduction; Theorem 2.7 in \cite{algo}]
    Given graphs $G(V,E)$ and $G'(V',E')$ with vertex $v \in V \setminus V'$ and $G'$ containing no isolated vertices,
    \begin{enumerate}[label=(\alph*)]
        \item If $v$ is a twin or leaf, then $G'<G \iff G'<G \setminus v$.
        \item If $v$ is an axil with some corresponding leaf $w$, then $G'<G \iff G'<(\tau_w \circ \tau_v(G) \setminus v)$.
    \end{enumerate}\label{theo:sreduction}
\end{theorem}

Notice that (b) is identical to (a) with the leaf and axil swapped. This result is a convenient way to reduce vertex-minor problems by removing vertices in the foliage.

In order to work with foliage graphs in the context of vertex-minors, we define a way to represent $F_W(G)$ on the same set of vertices as $G$.

\begin{definition}
    Let $G(V,E)$ be a graph with foliage partition $F_W(V)=\{W_1,\dots,W_k\}$ and foliage graph $F_W(G)$. Given a set of representatives $R = \{w_1,\dots,w_k\}$ with each $w_i \in W_i$, denote $F_{W,R}(G) \cong F_W(G)$ to be a \textbf{foliage graph labeling} with respect to $R$, where each vertex $W_i$ is labeled as $w_i$. 
\end{definition}

\begin{theorem}
    Given a graph $G(V,E)$ with foliage partition $F_W(V)=\{W_1,\dots,W_k\}$ and any choice of representatives $R = \{w_1,\dots,w_k\}$, we have $F_{W,R}(G) < G$. Succinctly, $F_W(G) < G$. \label{theo:fgraphvminor}
\end{theorem}

\begin{proof}
    %\footnote{While the foliage graph of each intermediate graph can change (ex. after the first step in Figure \ref{fig:fgraphvminor}, vertices 2 and 6 are twins), the key point is that the neighbors of $W_i$ are preserved after reduction to $w_i$.\label{foot:fgraphintermediate}}
    For each partition subset $W_i$, we define a process for reducing $W_i$ to $w_i$ without changing the overall foliage graph (see Figure \ref{fig:fgraphvminor}). For each type of $W_i$ given by Proposition 1 in ref. \cite{fpartition}:
    \begin{enumerate}[label=(\arabic*)]
        \item Single vertex: $W_i = \{w_i\}$ is already in the correct form.
        \item Star graph: Let $a\in W_i$ be the unique axil. Then $\tau_{w_i}\circ \tau_{a}(G)$ swaps the positions of $w_i$ and $a$ (which does nothing if $w_i$ is already the axil). The rest of the graph is unchanged, as the subgraph induced by $N_a \setminus W_i$ is complemented twice. Deleting every leaf of $w_i$, we obtain a single $w_i$ in place of $W_i$, where $N_{w_i}=N_a$.
        \item Fully connected graph: Delete every twin of $w_i$, which leaves the rest of the graph unchanged. The neighbors of vertices in $W_i$ are exactly the neighbors of $w_i$.
        \item Fully disconnected graph: Same as above.
    \end{enumerate}
    Since the edges between partition subsets are preserved, this process results in the foliage graph $F_{W,R}(G) < G$ on vertices $R=\{w_1,\dots,w_k\}$.
\end{proof}

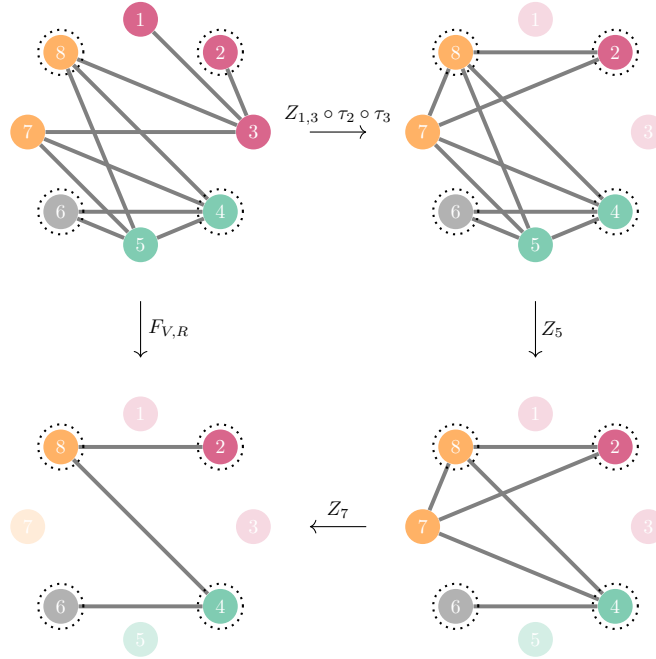
\begin{figure}[h]
\centering
\begin{tikzpicture}[scale=0.75, transform shape]
    %original graph
    \node[circle, fill=purple!60] (1) at (0,2) {\textcolor{white}{1}};
    \node[circle, fill=purple!60] (2) at (1.414,1.414) {\textcolor{white}{2}};
    \node[circle, fill=purple!60] (3) at (2,0) {\textcolor{white}{3}};
    \node[circle, fill=green!60!blue!50] (4) at (1.414,-1.414) {\textcolor{white}{4}};
    \node[circle, fill=green!60!blue!50] (5) at (0,-2) {\textcolor{white}{5}};
    \node[circle, fill=gray!60] (6) at (-1.414,-1.414) {\textcolor{white}{6}};
    \node[circle, fill=orange!60] (7) at (-2,0) {\textcolor{white}{7}};
    \node[circle, fill=orange!60] (8) at (-1.414,1.414) {\textcolor{white}{8}};
    
    \draw[gray, line width=1.5] (1) -- (3);
    \draw[gray, line width=1.5] (2) -- (3);
    \draw[gray, line width=1.5] (3) -- (7);
    \draw[gray, line width=1.5] (3) -- (8);
    \draw[gray, line width=1.5] (4) -- (5);
    \draw[gray, line width=1.5] (4) -- (6);
    \draw[gray, line width=1.5] (5) -- (6);
    \draw[gray, line width=1.5] (4) -- (7);
    \draw[gray, line width=1.5] (4) -- (8);
    \draw[gray, line width=1.5] (5) -- (7);
    \draw[gray, line width=1.5] (5) -- (8);
    
    \draw[dotted,thick] (2) circle (0.4);
    \draw[dotted,thick] (4) circle (0.4);
    \draw[dotted,thick] (6) circle (0.4);
    \draw[dotted,thick] (8) circle (0.4);
    
    %action on vertex 2
    \node[circle, fill=purple!15] (1) at (7,2) {\textcolor{white}{1}};
    \node[circle, fill=purple!60] (2) at (8.414,1.414) {\textcolor{white}{2}};
    \node[circle, fill=purple!15] (3) at (9,0) {\textcolor{white}{3}};
    \node[circle, fill=green!60!blue!50] (4) at (8.414,-1.414) {\textcolor{white}{4}};
    \node[circle, fill=green!60!blue!50] (5) at (7,-2) {\textcolor{white}{5}};
    \node[circle, fill=gray!60] (6) at (5.586,-1.414) {\textcolor{white}{6}};
    \node[circle, fill=orange!60] (7) at (5,0) {\textcolor{white}{7}};
    \node[circle, fill=orange!60] (8) at (5.586,1.414) {\textcolor{white}{8}};
    
    \draw[gray, line width=1.5] (2) -- (7);
    \draw[gray, line width=1.5] (2) -- (8);
    \draw[gray, line width=1.5] (4) -- (5);
    \draw[gray, line width=1.5] (4) -- (6);
    \draw[gray, line width=1.5] (5) -- (6);
    \draw[gray, line width=1.5] (4) -- (7);
    \draw[gray, line width=1.5] (4) -- (8);
    \draw[gray, line width=1.5] (5) -- (7);
    \draw[gray, line width=1.5] (5) -- (8);
    \draw[gray, line width=1.5] (7) -- (8);
    
    \draw[dotted,thick] (2) circle (0.4);
    \draw[dotted,thick] (4) circle (0.4);
    \draw[dotted,thick] (6) circle (0.4);
    \draw[dotted,thick] (8) circle (0.4);
    
    %action on vertex 4
    \node[circle, fill=purple!15] (1) at (7,-5) {\textcolor{white}{1}};
    \node[circle, fill=purple!60] (2) at (8.414,-5.586) {\textcolor{white}{2}};
    \node[circle, fill=purple!15] (3) at (9,-7) {\textcolor{white}{3}};
    \node[circle, fill=green!60!blue!50] (4) at (8.414,-8.414) {\textcolor{white}{4}};
    \node[circle, fill=green!60!blue!17] (5) at (7,-9) {\textcolor{white}{5}};
    \node[circle, fill=gray!60] (6) at (5.586,-8.414) {\textcolor{white}{6}};
    \node[circle, fill=orange!60] (7) at (5,-7) {\textcolor{white}{7}};
    \node[circle, fill=orange!60] (8) at (5.586,-5.586) {\textcolor{white}{8}};
    
    \draw[gray, line width=1.5] (2) -- (7);
    \draw[gray, line width=1.5] (2) -- (8);
    \draw[gray, line width=1.5] (4) -- (6);
    \draw[gray, line width=1.5] (4) -- (7);
    \draw[gray, line width=1.5] (4) -- (8);
    \draw[gray, line width=1.5] (7) -- (8);
    
    \draw[dotted,thick] (2) circle (0.4);
    \draw[dotted,thick] (4) circle (0.4);
    \draw[dotted,thick] (6) circle (0.4);
    \draw[dotted,thick] (8) circle (0.4);
    
    %action on vertex 8
    \node[circle, fill=purple!15] (1) at (0,-5) {\textcolor{white}{1}};
    \node[circle, fill=purple!60] (2) at (1.414,-5.586) {\textcolor{white}{2}};
    \node[circle, fill=purple!15] (3) at (2,-7) {\textcolor{white}{3}};
    \node[circle, fill=green!60!blue!50] (4) at (1.414,-8.414) {\textcolor{white}{4}};
    \node[circle, fill=green!60!blue!17] (5) at (0,-9) {\textcolor{white}{5}};
    \node[circle, fill=gray!60] (6) at (-1.414,-8.414) {\textcolor{white}{6}};
    \node[circle, fill=orange!15] (7) at (-2,-7) {\textcolor{white}{7}};
    \node[circle, fill=orange!60] (8) at (-1.414,-5.586) {\textcolor{white}{8}};
    
    \draw[gray, line width=1.5] (2) -- (8);
    \draw[gray, line width=1.5] (4) -- (6);
    \draw[gray, line width=1.5] (4) -- (8);
    
    \draw[dotted,thick] (2) circle (0.4);
    \draw[dotted,thick] (4) circle (0.4);
    \draw[dotted,thick] (6) circle (0.4);
    \draw[dotted,thick] (8) circle (0.4);
    
    \draw[->] (3,0)--(4,0);
    \node at (3.5,0.3) {$Z_{1,3} \circ \tau_2 \circ \tau_3$};
    \draw[->] (7,-3)--(7,-4);
    \node at (7.3,-3.5) {$Z_5$};
    \draw[->] (4,-7)--(3,-7);
    \node at (3.5,-6.7) {$Z_7$};
    \draw[->] (0,-3)--(0,-4);
    \node at (0.5,-3.5) {$F_{V,R}$};

\end{tikzpicture}
\caption{Process in Theorem \ref{theo:fgraphvminor} illustrated for extracting $F_{V,R}(G)$ from an eight vertex graph, with representatives $R=\{2,4,6,8\}$ circled. Actions on partition subsets are shown for a star graph (vertices 1, 2, 3), fully disconnected graph (vertices 4, 5) and fully connected graph (vertices 7, 8). Nothing is done for single vertex 6.}
\label{fig:fgraphvminor}
\end{figure}

Theorem \ref{theo:fgraphvminor} leads to a more general formulation of source reduction, which reveals the underlying foliage structure.

\begin{theorem}[Foliage Source Reduction]
    Let $G(V,E)$ and $G'(V',E')$ be graphs with $F_{W,R}(G)$ defined on labels $R$ and $V'\subset R \subset V$. If $G'$ has no isolated vertices, then $G'<G \iff G'<F_{W,R}(G)$.
\end{theorem}
\begin{proof}
    Since the process in the proof of Theorem \ref{theo:fgraphvminor} only involves LCs (clearly, $G'<G \iff G'<\tau(G)$) and deletions of leaves and twins, we can apply Theorem \ref{theo:sreduction}. We require $V' \subset R$ to ensure that vertices in the target graph are not deleted in the above process.
\end{proof}

\begin{lemma}
    Given graphs $G'<G$ with a set of pairwise foliage-equivalent vertices $A \subset G$, the vertices in $A \cap G'$ are either pairwise foliage-equivalent to each other, all isolated vertices, or the empty set.\label{lemma:fgec}
\end{lemma}
\begin{proof}
    We show that any LC or vertex deletion can only send the vertices in $W_i$ to one of the three descriptions.

    As in Theorem \ref{theo:fplci}, LCs on a set of equivalent vertices leaves them equivalent. LCs cannot change a set of isolated vertices or the empty set. Vertex deletion inside a set of equivalent vertices either 1) results in a set of isolated vertices, if an axil is deleted, 2) results in the empty set, if there is 1 vertex, or 3) does nothing, otherwise. Vertex deletion in a neighbor of this set either 1) results in a set of isolated vertices, if the set of equivalent vertices is totally disconnected and share only one neighbor, or 2) does nothing, otherwise.

    So inductively, the vertices in $A \cup G'$ follow one of the three descriptions.
\end{proof}

While sets of equivalent vertices maintain some relation under vertex deletion, nonequivalent vertices may become equivalent (consider deletion at vertex 1 on a line graph; vertices 2 and 3 become equivalent). While this lemma shows that much of the structure of the foliage graph is retained, a vertex-minor graph can have a different canonical foliage graph.

\subsection{Target Reduction Theorem}

Source reduction only allows us to delete vertices not in the target graph. We can allow for the other case through a process we call \textit{target reduction}. This reduction turns out to be more restrictive.

\begin{theorem}[Target Reduction]
    Let $G(V,E)$ and $G'(V',E')$ be graphs with $G' < G$ and vertices $v \sim_\text{F} w$ in both $G$ and $G'$.
    \begin{enumerate}[label=(\arabic*)]
        \item If $(v,w) \in L_G$ or $\{v,w\} \in T_G$, let $\tilde{G} = G \setminus v$.
        
        \item If $(w,v) \in L_G$, let $\tilde{G} = \tau_w \circ \tau_v(G) \setminus v$.
        
        \item If $(v,w) \in L_{G'}$ or $\{v,w\} \in T_{G'}$, let $\tilde{G'} = G' \setminus v$.
        
        \item If $(w,v) \in L_{G'}$, let $\tilde{G'} = \tau_w \circ \tau_v(G') \setminus v$.
    \end{enumerate}
    Then $\tilde{G'} < \tilde{G}$.
    \label{theo:treduction}
\end{theorem}

\begin{proof}
    From $G'<G$ we have a sequence of local complementations and vertex deletions $f_1, \dots, f_m$ with $f_m \circ \dots \circ f_1 (G) = G'$ yielding intermediate graphs $G=G_0, G_1, \dots, G_m=G'$, as shown below.
    
    We use the graph operation $F_{W_l,R_l}$ to send each $G_l$ to a $G'_l$, where $W_l$ contains every singleton except $\{v,w\} \in W_l$, and $R_l = V(G_l) \setminus v$. To show that this is well-defined, we need $v \sim_{\text{F}} w$ for every intermediate graph. We use Lemma \ref{lemma:fgec} inductively beginning with $G_1 < G$ and $A = \{v,w\}$. Then $A \cap G_1 = \{v,w\}$ since $v$ and $w$ can't have been deleted. So $v$ and $w$ must either be foliage-equivalent or isolated vertices. But since $G, G_1, \dots, G'$ is a sequence resulting in $v \sim_\text{F} w$ in $G'$, they cannot become isolated. So $v \sim_\text{F} w$ in $G_1$.

    \begin{center}
        \begin{tikzcd}[column sep=1.5cm, row sep=1.5cm]
        G \arrow{r}{f_1} \arrow[swap]{d}{F_{W_0,R_0}} & G_1 \arrow{r}{f_2} \arrow[swap]{d}{F_{W_1,R_1}} & G_2 \arrow{r}{f_3} \arrow[swap]{d}{F_{W_2,R_2}} & \cdots \arrow{r}{f_m} & G' \arrow[swap]{d}{F_{W_m,R_m}} \\
        \tilde{G} \arrow{r}{g_1} & G'_1 \arrow{r}{g_2} & G'_2 \arrow{r}{g_3} & \cdots \arrow{r}{g_m} & \tilde{G'}
    \end{tikzcd}
    \end{center}

    We have $G'_0 = \tilde{G}$ and $G'_m = \tilde{G'}$, since $F_{W_0,R_0}$ and $F_{W_m,R_m}$ correspond to the deletions described in the theorem statement.
    
    We assert that the sequence $\tilde{G} = G'_0, G'_1, \dots, G'_m = \tilde{G'}$ is obtained by local complementations and vertex deletions. For each $f_l$ we define $g_l$ on this sequence as follows (in each case, $a \in V_i$):
    \begin{align}
        g_l=\begin{cases}
            \tau_a & \text{if } f_l=\tau_a, a \neq v, \abs{N_a^{(G_{l-1})}}>1. \\
            \tau_w & \text{if } f_l=\tau_v, \abs{N_v^{(G_{l-1})}}>1. \\
            Z_a & \text{if } f_l = Z_a. \\
            I &\text{otherwise.}
        \end{cases}\label{eq:repops}
    \end{align}

    The first two essentially describe lifted local complementation. The composition of all $g_l$ is a sequence of local complementations and vertex deletions. We want to show that $g_l$ do indeed send $G'_{l-1}$ to $G'_l$. That is, we want $g_l\circ F_{W_{l-1},R_{l-1}}(G_{l-1}) = F_{W_l,R_l} \circ f_l (G_{l-1})$. Going through each operation in \eqref{eq:repops},
    
    \begin{enumerate}[label=(\arabic*)]
        \item $\tau_a \circ F_{W_{l-1},R_{l-1}}(G_{l-1}) \stackrel{?}{=} F_{W_l,R_l} \circ \tau_a (G_{l-1})$: Since no vertices are deleted by $\tau_a$, $F_{W_{l-1},R_{l-1}} = F_{W_l,R_l}$. So equality holds by Theorem \ref{theo:liftedlc}, since each $a$ corresponds to itself in the foliage graph.
        
        \item $\tau_w \circ F_{W_{l-1},R_{l-1}}(G_{l-1}) \stackrel{?}{=} F_{W_l,R_l} \circ \tau_v (G_{l-1})$: Same as above; $v$ corresponds to $w$ in the foliage graph.
        
        \item $Z_a \circ F_{W_{l-1},R_{l-1}}(G_{l-1}) \stackrel{?}{=} F_{W_l,R_l} \circ Z_a(G_{l-1})$: Since the foliage graph reduction only changes $v$, $w$, and incident edges, we only have to check when $a$ is a neighbor of $v$ or $w$. If $v$ is a leaf of $w$ or they are twins, then in the LHS, $N_w^{(G'_l)} = N_w^{(G_{l-1})} \setminus v \setminus a$, and in the RHS, $N_w^{(G'_l)} = N_w^{(G_{l-1})} \setminus a \setminus v$. If $w$ is a leaf of $v$, then in the LHS, $N_w^{(G'_l)} = N_v^{(G_{l-1})} \setminus w \setminus a$, and in the RHS, $N_w^{(G'_l)} = N_v^{(G_{l-1})} \setminus a \setminus w$.
        
        \item $F_{W_{l-1},R_{l-1}}(G_{l-1}) \stackrel{?}{=} F_{W_l,R_l} (G_{l-1})$: The only remaining case is if $f_l = \tau_a$ and $\abs{N_a} \leq 1$. Then $F_{W_{l-1},R_{l-1}} = F_{W_l,R_l}$ and equality holds.
    \end{enumerate}

    So $g_1, \dots, g_m$ is a sequence of local complementations and vertex deletions such that $g_m \circ \dots \circ g_1 (\tilde{G}) = \tilde{G'}$, yielding $\tilde{G'}<\tilde{G}$.
\end{proof}

\begin{figure}[h]
\begin{center}
\begin{tikzpicture}[scale=0.75, transform shape]
    %original
    \node[circle, fill=purple!60, inner sep=2pt] (1) at (-2,-6.5) {\textcolor{white}{$a_1$}};
    \node[circle, fill=purple!60, inner sep=2pt] (2) at (-0.5,-6.5) {\textcolor{white}{$a_2$}};
    \node[circle, fill=purple!60, minimum size=17pt] (3) at (1,-6.5) {};
    \node[circle, fill=orange!60, inner sep=1.6pt] (4) at (2.5,-5.75) {\textcolor{white}{$b_1$}};
    \node[circle, fill=orange!60, inner sep=1.6pt] (5) at (2.5,-7.25) {\textcolor{white}{$b_2$}};

    \draw[gray, line width=1.5] (1) -- (2);
    \draw[gray, line width=1.5] (2) -- (3);
    \draw[gray, line width=1.5] (3) -- (4);
    \draw[gray, line width=1.5] (3) -- (5);

    \draw[gray!50, line width=1, <->] (-1.25,-8) +(75:2cm) arc (75:105:2cm);
    \draw[gray!50, line width=1, <->] (1,-6.5) +(-15:2cm) arc (-15:15:2cm);

    \filldraw[gray!10] (-4,-6.5) circle (1.1);
    \node[red,scale=3] at (-4,-6) {\ding{55}};
    \node[circle, fill=purple!60, inner sep=2pt] (1) at (-4.25,-6.75) {};
    \node[circle, fill=purple!60, inner sep=2pt] (2) at (-4.25,-7.25) {};
    \node[circle, fill=orange!60, inner sep=2pt] (3) at (-3.75,-6.75) {};
    \node[circle, fill=orange!60, inner sep=2pt] (4) at (-3.75,-7.25) {};
    \draw[gray, line width=1.5] (1) -- (2);
    \draw[gray, line width=1.5] (3) -- (4);

    %target-reduced
    \node[circle, fill=purple!60, inner sep=2pt] (1) at (-2,-9.5) {\textcolor{white}{$a_1$}};
    \node[circle, fill=purple!60, inner sep=2pt] (2) at (-0.5,-9.5) {\textcolor{white}{$a_2$}};
    \node[circle, fill=purple!60, minimum size=17pt] (3) at (1,-9.5) {};
    \node[circle, fill=orange!60, inner sep=1.6pt] (4) at (2.5,-9.5) {\textcolor{white}{$b_1$}};

    \draw[gray, line width=1.5] (1) -- (2);
    \draw[gray, line width=1.5] (2) -- (3);
    \draw[gray, line width=1.5] (3) -- (4);

    \draw[gray!50, line width=1, <->] (-1.25,-11) +(75:2cm) arc (75:105:2cm);
    \draw[gray!50, line width=1, ->] (2.5,-8.6) to (2.5,-9.1);

    \filldraw[gray!10] (-4,-9.5) circle (1.1);
    \node[green!70!blue!70,scale=3] at (-4,-9) {\ding{51}};
    \node[circle, fill=purple!60, inner sep=2pt] (1) at (-4.25,-9.75) {};
    \node[circle, fill=purple!60, inner sep=2pt] (2) at (-4.25,-10.25) {};
    \node[circle, fill=orange!60, inner sep=2pt] (3) at (-3.75,-10) {};
    \draw[gray, line width=1.5] (1) -- (2);

    \draw[->] (0.25,-7.5) -- (0.25,-8.5);
    \node at (0.5,-8) {T};

    %original
    \node[circle, fill=purple!60, inner sep=2pt] (1) at (9,-3.5) {\textcolor{white}{$a_1$}};
    \node[circle, fill=purple!60, inner sep=2pt] (2) at (10.5,-3.5) {\textcolor{white}{$a_2$}};
    \node[circle, fill=purple!60, minimum size=17pt] (3) at (12,-3.5) {};
    \node[circle, fill=orange!60, inner sep=1.6pt] (4) at (13.5,-3.5) {\textcolor{white}{$b_1$}};
    \node[circle, fill=orange!60, inner sep=1.6pt] (5) at (15,-3.5) {\textcolor{white}{$b_2$}};
    
    \draw[gray, line width=1.5] (1) -- (2);
    \draw[gray, line width=1.5] (2) -- (3);
    \draw[gray, line width=1.5] (3) -- (4);
    \draw[gray, line width=1.5] (4) -- (5);
    
    \draw[gray!50, line width=1, <->] (9.75,-5)+(75:2cm) arc (75:105:2cm);
    \draw[gray!50, line width=1, <->] (14.25,-5)+(75:2cm) arc (75:105:2cm);
    
    \filldraw[gray!10] (7,-3.5) circle (1.1);
    \node[green!70!blue!70,scale=3] at (7,-3) {\ding{51}};
    \node[circle, fill=purple!60, inner sep=2pt] (1) at (6.75,-3.75) {};
    \node[circle, fill=purple!60, inner sep=2pt] (2) at (6.75,-4.25) {};
    \node[circle, fill=orange!60, inner sep=2pt] (3) at (7.25,-3.75) {};
    \node[circle, fill=orange!60, inner sep=2pt] (4) at (7.25,-4.25) {};
    \draw[gray, line width=1.5] (1) -- (2);
    \draw[gray, line width=1.5] (3) -- (4);

    % target-reduced
    \node[circle, fill=purple!60, inner sep=2pt] (1) at (9,-6.5) {\textcolor{white}{$a_1$}};
    \node[circle, fill=purple!60, inner sep=2pt] (2) at (10.5,-6.5) {\textcolor{white}{$a_2$}};
    \node[circle, fill=purple!60, minimum size=17pt] (3) at (12,-6.5) {};
    \node[circle, fill=orange!60, inner sep=1.6pt] (4) at (13.5,-6.5) {\textcolor{white}{$b_1$}};
    
    \draw[gray, line width=1.5] (1) -- (2);
    \draw[gray, line width=1.5] (2) -- (3);
    \draw[gray, line width=1.5] (3) -- (4);

    \draw[gray!50, line width=1, <->] (9.75,-8)+(75:2cm) arc (75:105:2cm);
    \draw[gray!50, line width=1, ->] (13.5,-5.6) to (13.5,-6.1);
    
    \filldraw[gray!10] (7,-6.5) circle (1.1);
    \node[green!70!blue!70,scale=3] at (7,-6) {\ding{51}};
    \node[circle, fill=purple!60, inner sep=2pt] (1) at (6.75,-6.75) {};
    \node[circle, fill=purple!60, inner sep=2pt] (2) at (6.75,-7.25) {};
    \node[circle, fill=orange!60, inner sep=2pt] (3) at (7.25,-7) {};
    \draw[gray, line width=1.5] (1) -- (2);

    \draw[->] (11.25,-4.5) -- (11.25,-5.5);
    \node at (11.5,-5) {T};

    % source-reduced
    \node[circle, fill=purple!60, inner sep=2pt] (1) at (9,-9.5) {\textcolor{white}{$a_1$}};
    \node[circle, fill=purple!60, inner sep=2pt] (2) at (10.5,-9.5) {\textcolor{white}{$a_2$}};
    \node[circle, fill=orange!60, inner sep=1.6pt] (3) at (12,-9.5) {\textcolor{white}{$b_1$}};
    
    \draw[gray, line width=1.5] (1) -- (2);
    \draw[gray, line width=1.5] (2) -- (3);

    \draw[gray!50, line width=1, <->] (9.75,-11)+(75:2cm) arc (75:105:2cm);
    \draw[gray!50, line width=1, ->] (12,-8.6) to (12,-9.1);
    
    \filldraw[gray!10] (7,-9.5) circle (1.1);
    \node[red,scale=3] at (7,-9) {\ding{55}};
    \node[circle, fill=purple!60, inner sep=2pt] (1) at (6.75,-9.75) {};
    \node[circle, fill=purple!60, inner sep=2pt] (2) at (6.75,-10.25) {};
    \node[circle, fill=orange!60, inner sep=2pt] (3) at (7.25,-10) {};
    \draw[gray, line width=1.5] (1) -- (2);

    \draw[->] (11.25,-7.5) -- (11.25,-8.5);
    \node at (11.5,-8) {S};

    \begin{scope}[scale=1.333,transform shape]
      \node at (-4.5,-7.75) {(a)};
      \node at (3.75,-7.75) {(b)};
    \end{scope}
\end{tikzpicture}
\end{center}
\caption{Limitations of target reduction demonstrated using Bell pair extraction (see Section \ref{sec:bellvm}). (a) Unlike in source reduction, the inverse of target reduction does not hold. We have (top) $K_2 \sqcup K_2 \nless G$, but target reduction with $b_1 \sim_\text{F} b_2$ gives (bottom) $K_2 \sqcup \{b_1\} < G \setminus \{b_2\}$. (b) Target reduction may introduce isolated vertices into the target graph, disallowing source reduction. We have (top) $K_2 \sqcup K_2 < G$, and target reduction with $b_1 \sim_\text{F} b_2$ gives (middle) $K_2 \sqcup \{b_1\} < G \setminus \{b_2\}$. But then source reduction on the unlabeled vertex gives (bottom) an impossible three-vertex extraction.}
\label{fig:trlimits}
\end{figure}
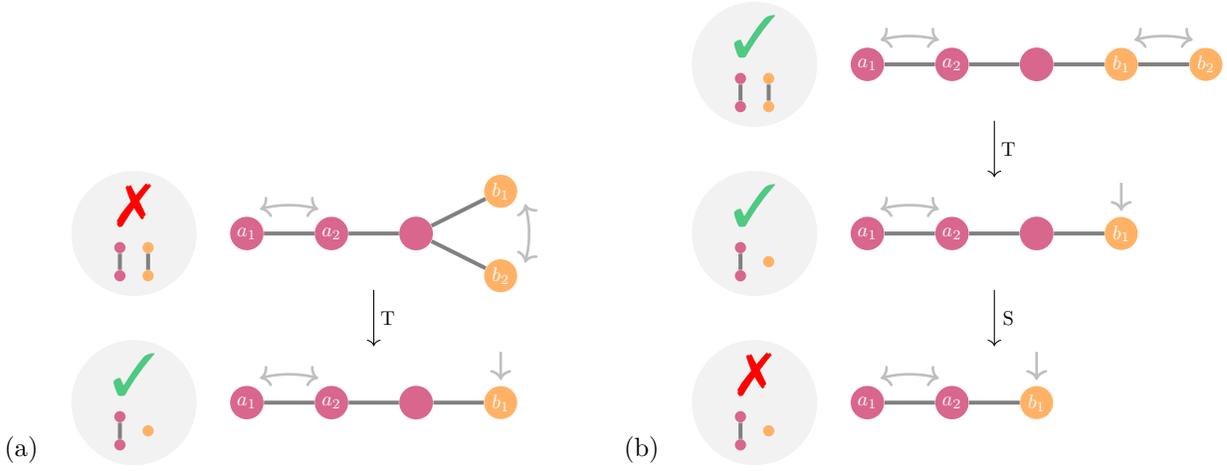

We can define foliage target reduction analogously to foliage source reduction.

\begin{theorem}[Foliage Target Reduction]
    Let $G(V,E)$ and $G'(V',E')$ be graphs with $G'<G$ and $F_{W',R'}(G')$ defined on labels $R'$. Let $F_W(V) = F_{W'}(V') \cup \{ \{v\} \mid v \in V \setminus V' \}$ and $R = R' \cup (V \setminus V')$. If $F_W(V)$ is a foliage partition, then $F_{W',R'}(G')<F_{W,R}(G)$.
\end{theorem}
\begin{proof}
    Since the process in the proof of Theorem \ref{theo:fgraphvminor} for $F_{W,R}(G) < G$ only involves LCs and deletions of leaves and twins in $V'$ (all other partition subsets are singletons in $V \setminus V'$), we can apply Theorem \ref{theo:treduction} to obtain $F_{W',R'}(G')<F_{W,R}(G)$ from $G'<G$.
    
    It remains to show that the actions on $G'$ do yield $F_{W',R'}(G')$. First note that the partition subsets of $F_{W'}(V')$ remain sets of equivalent vertices under (3) and (4) of Theorem \ref{theo:treduction}. By the proof of Lemma \ref{lemma:fgec}, a partition subset can only become nonequivalent if an axil is deleted. However, only leaves and twins are deleted in Theorem \ref{theo:treduction}.

    Now consider each partition subset $W'_i$ in $F_{W'}(V')$. For each type of $W'_i$ given by Proposition 1 in ref. \cite{fpartition}:
    \begin{enumerate}[label=(\arabic*)]
        \item Single vertex: Already the correct form; nothing is done to the corresponding partition $W_i=W'_i$ in $G$.
        \item Star graph: If the representative element $w'_i$ is the axil, then all of the leaves are deleted by (3) of Theorem \ref{theo:treduction}. Otherwise, when the axil is deleted by (4) of Theorem \ref{theo:treduction}, $w'_i$ becomes the axil, yielding the same result.
        \item Fully connected graph: All vertices except the representative element are deleted.
        \item Fully disconnected graph: Same as above.
    \end{enumerate}
    Since the edges between foliage partition subsets are preserved, this process results in the foliage graph $F_{W',R'}(G')$ on vertices $R'$.
\end{proof}

In short, foliage target reduction allows us to reduce sets of equivalent vertices which are also equivalent in the target graph.

\section{Bell Vertex-Minor}\label{sec:bellvm}

To determine whether simultaneous extraction of two Bell pairs is possible, we want to know if $K_2\sqcup K_2 = G'(\{a_1,a_2,b_1,b_2\}, \{\{a_1,a_2\},\{b_1,b_2\}\}) < G$ given graph $G$ and vertices $a_1,a_2,b_1,b_2 \in G$. In this section we use our tools, target reduction in particular, to solve the \textbf{Bell vertex-minor problem} on line, tree, and ring graphs.

\subsection{Line and Tree Graphs}

The \textbf{line graph} is defined as $L_n(V_n, E_n)$ with $V_n = \{1,\dots, n\}$ and $E_n=\{\{1,2\}, \{2,3\}, \dots, \{n-1,n\}\}$. Theorem 3 in ref. \cite{limits} showed that $G(\{a_1,b_1,a_2,b_2\}, \{\{a_1,a_2\},\{b_1,b_2\}\}) \nless L_n$ when $a_1<b_1<a_2<b_2$. This result is one of three types of bottlenecks (defined simply as when extraction is impossible) on line graphs we will identify. We will solve the line graph completely and generalize the result to trees.
 
\begin{lemma}
    It is not possible to extract a Bell pair $\{a_1,a_2\}$ and an isolated\footnote{Referring to an isolated, unmeasured qubit.} qubit $b$ from $L_n$ if $a_1 < b < a_2$.\label{lemma:line1a}
\end{lemma}

\begin{proof}
    We proceed through strong induction as in the $a_1<b_1<a_2<b_2$ case. The base case for $n = 3$ is trivial. Assuming the result holds up to some $n\geq 3$, we want to show that $K_2 \sqcup \{b\} \nless L_{n+1}$. By the contrapositive of Theorem \ref{theo:preduction}, this is equivalent to showing that $K_2 \sqcup \{b\} \nless P_v(L_{n+1})$ for each $P_v \in \{X_v,Y_v,Z_v\}$.
    \begin{enumerate}[label=(\arabic*)]
        \item $K_2 \sqcup \{b\} \stackrel{?}{\nless} Z_v(L_{n+1})$: If $v<a_1$ or $v>a_2$, then we apply the inductive hypothesis on the component containing $a_1$, $b$, and $a_2$. Otherwise, $Z_v$ disconnects $a_1$ and $a_2$. Then $K_2 \sqcup \{b\} \nless Z_v(L_{n+1})$, since local complementations (and clearly, vertex deletions) can't reconnect disconnected components.
        \item $K_2 \sqcup \{b\} \stackrel{?}{\nless} Y_v(L_{n+1})$: The graph becomes $L_n$ with the order of $a_1$, $b$, and $a_2$ unchanged, so we can apply the inductive hypothesis.
        \item $K_2 \sqcup \{b\} \stackrel{?}{\nless} X_v(L_{n+1})$: The graph becomes $L_n$ except the two neighbors of $v$ are a leaf-axil pair\footnote{\label{note:xmeas}The order depends on the choice of neighbor in $X_v$; LCs can freely reverse the order}. If either of this pair is not $a_1$, $b$, or $a_2$, then we can apply Theorem \ref{theo:sreduction} followed by the inductive hypothesis. If this pair is $\{a_1,b\}$ (we cannot have $\{a_1,a_2\}$ since $a_1<b<a_2$ and by symmetry $\{b,a_2\}$ is the same), then $a_1 \sim_\text{F} b$. Lemma \ref{lemma:fgec} yields a contradiction, since none of the descriptions correspond to $\{a_1,b\}$ in $K_2 \sqcup \{b\}$.
    \end{enumerate}
    So the result holds for all $n\geq 3$.
\end{proof}

%insert lemma ring1a

\begin{proposition}
    It is not possible to extract two Bell pairs $\{a_1,a_2\}$ and $\{b_1,b_2\}$ from $L_n$ if $a_1 < b_1 < b_2 < a_2$.\label{prop:line2}
\end{proposition}

\begin{proof}
    We use induction as before. The base case for $n = 4$ is trivial. Assuming the result holds up to some $n\geq 4$, again we want to show that $K_2 \sqcup K_2 \nless P_v(L_{n+1})$ for each $P_v \in \{X_v,Y_v,Z_v\}$.
    \begin{enumerate}[label=(\arabic*)]
        \item $K_2 \sqcup K_2 \stackrel{?}{\nless} Z_v(L_{n+1})$: Same as Lemma \ref{lemma:line1a}.
        \item $K_2 \sqcup K_2 \stackrel{?}{\nless} Y_v(L_{n+1})$: Same as Lemma \ref{lemma:line1a}.
        \item $K_2 \sqcup K_2 \stackrel{?}{\nless} X_v(L_{n+1})$: The graph becomes $L_n$ except the two neighbors of $v$ are a leaf-axil pair. This case proceeds as in Lemma \ref{lemma:line1a} except when the pair is $\{b_1,b_2\}$. In this case, $b_1 \sim_\text{F} b_2$ in both $X_v(L_{n+1})$ and $K_2 \sqcup K_2$. So we can apply Theorem \ref{theo:treduction} with leaf deletion to get $K_2 \sqcup K_2 < X_v(L_{n+1}) \implies K_2 \sqcup \{b\} < L_n$. This reduces to the problem in Lemma \ref{lemma:line1a}, where we know extraction is impossible. So $K_2 \sqcup K_2 \nless X_v(L_{n+1})$.
    \end{enumerate}
    So the result holds for all $n\geq 4$.
\end{proof}

\begin{figure}[h]
\begin{center}
\begin{tikzpicture}[scale=0.75, transform shape]
  %Original graph
  \node[circle, fill=purple!60] (1) at (0,0) {\textcolor{white}{1}};
  \node[circle, fill=purple!60, minimum size=17pt] (2) at (1.5,0) {\textcolor{white}{}};
  \node[circle, fill=purple!60] (3) at (3,0) {\textcolor{white}{$u$}};
  \node[circle, fill=purple!80] (4) at (4.5,0) {\textcolor{white}{$v$}};
  \node[circle, fill=purple!60] (5) at (6,0) {\textcolor{white}{$w$}};
  \node[circle, fill=purple!60, minimum size=17pt] (6) at (7.5,0) {\textcolor{white}{}};
  \node[circle, fill=purple!60, inner sep=0pt] (7) at (9,0) {\textcolor{white}{\footnotesize$n \kern-1.4pt + \kern-1.4pt 1$}};
  
  \draw[gray, line width=1.5] (2) -- (3);
  \draw[gray, line width=1.5] (3) -- (4);
  \draw[gray, line width=1.5] (4) -- (5);
  \draw[gray, line width=1.5] (5) -- (6);
  \node at (0.75,0) {$\dots$};
  \node at (8.25,0) {$\dots$};

  %Z_v
  \node[circle, fill=purple!60] (1) at (0,-2) {\textcolor{white}{1}};
  \node[circle, fill=purple!60, minimum size=17pt] (2) at (1.5,-2) {\textcolor{white}{}};
  \node[circle, fill=purple!60] (3) at (3,-2) {\textcolor{white}{$u$}};
  \node[circle, fill=purple!60] (5) at (4.5,-2) {\textcolor{white}{$w$}};
  \node[circle, fill=purple!60, minimum size=17pt] (6) at (6,-2) {\textcolor{white}{}};
  \node[circle, fill=purple!60, inner sep=0pt] (7) at (7.5,-2) {\textcolor{white}{\footnotesize$n \kern-1.4pt + \kern-1.4pt 1$}};
  
  \draw[gray, line width=1.5] (2) -- (3);
  \draw[gray, line width=1.5] (5) -- (6);
  \node at (0.75,-2) {$\dots$};
  \node at (6.75,-2) {$\dots$};
  
  %Y_v
  \node[circle, fill=purple!60] (1) at (0,-4) {\textcolor{white}{1}};
  \node[circle, fill=purple!60, minimum size=17pt] (2) at (1.5,-4) {\textcolor{white}{}};
  \node[circle, fill=purple!60] (3) at (3,-4) {\textcolor{white}{$u$}};
  \node[circle, fill=purple!60] (5) at (4.5,-4) {\textcolor{white}{$w$}};
  \node[circle, fill=purple!60, minimum size=17pt] (6) at (6,-4) {\textcolor{white}{}};
  \node[circle, fill=purple!60, inner sep=0pt] (7) at (7.5,-4) {\textcolor{white}{\footnotesize$n \kern-1.4pt + \kern-1.4pt 1$}};
  
  \draw[gray, line width=1.5] (2) -- (3);
  \draw[gray, line width=1.5] (5) -- (6);
  \draw[gray, line width=1.5] (3) -- (5);
  \node at (0.75,-4) {$\dots$};
  \node at (6.75,-4) {$\dots$};
  
  %X_v
  \node[circle, fill=purple!60] (1) at (0,-6) {\textcolor{white}{1}};
  \node[circle, fill=purple!60, minimum size=17pt] (2) at (1.5,-6) {\textcolor{white}{}};
  \node[circle, fill=purple!60] (3) at (3,-6) {\textcolor{white}{$u$}};
  \node[circle, fill=purple!60] (5) at (3,-7.5) {\textcolor{white}{$w$}};
  \node[circle, fill=purple!60, minimum size=17pt] (6) at (4.5,-6) {\textcolor{white}{}};
  \node[circle, fill=purple!60, inner sep=0pt] (7) at (6,-6) {\textcolor{white}{\footnotesize$n \kern-1.4pt + \kern-1.4pt 1$}};
  
  \draw[gray, line width=1.5] (2) -- (3);
  \draw[gray, line width=1.5] (3) -- (6);
  \draw[gray, line width=1.5] (3) -- (5);
  \node at (0.75,-6) {$\dots$};
  \node at (5.25,-6) {$\dots$};

  \begin{scope}[scale=1.3333, transform shape]
    \node[align=left, text width=5cm] at (-1,0) {(a) $L_{n+1}$};
    \node[align=left, text width=5cm] at (-1,-1.5) {(b) $Z_v(L_{n+1})$};
    \node[align=left, text width=5cm] at (-1,-3) {(c) $Y_v(L_{n+1})$};
    \node[align=left, text width=5cm] at (-1,-4.5) {(d) $X_v(L_{n+1})$};
  \end{scope}
  
\end{tikzpicture}
\end{center}
\caption{(a)-(d) Pauli measurements on a line graph on vertex $v$.}
\label{fig:line2}
\end{figure}
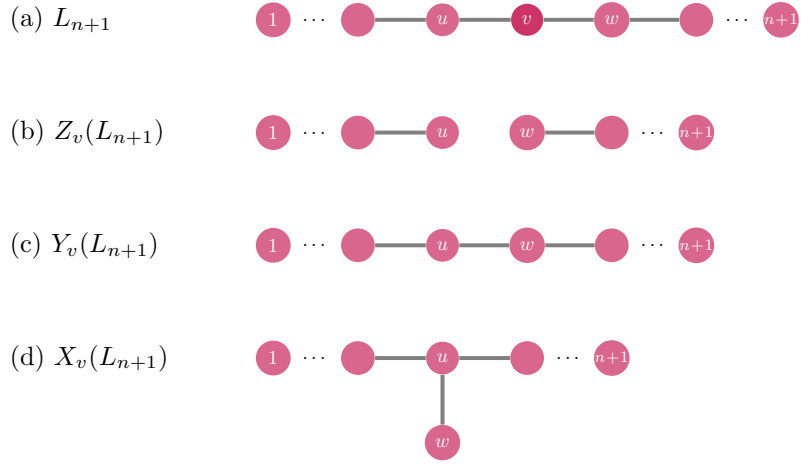

\begin{proposition}
    It is not possible to extract two Bell pairs $\{a_1,a_2\}$ and $\{b_1,b_2\}$ from $L_n$ if $a_1 < a_2 < b_1 < b_2$ and $a_2$ and $b_1$ are adjacent.\label{prop:line3}
\end{proposition}
\begin{proof}
    By repeated application of Theorem \ref{theo:sreduction}a on either end of the graph, we obtain a line graph with $a_1$ and $b_2$ as the endpoints. Then by repeated application of Theorem \ref{theo:sreduction}b on either end, we obtain the line graph $L_4$, from which extraction of $K_2 \sqcup K_2$ is impossible. Since this extraction is equivalent to the original extraction problem, the original is also impossible.
\end{proof}

These two propositions along with Theorem 3 in \cite{limits} allow us to solve the Bell vertex-minor problem completely for line graphs.

\begin{theorem}
    Two Bell pairs $\{a_1,a_2\}$ and $\{b_1,b_2\}$ can be extracted from $L_n$ iff $a_1 < a_2 < b_1 < b_2$ and $a_2$ and $b_1$ are not adjacent.\label{theo:linebellvm}
\end{theorem}
\begin{proof}
    Consider the contrapositive of the forwards direction. All arrangements not described in the statement fall under one of the Proposition \ref{prop:line2}, Proposition \ref{prop:line3}, or Theorem 3 in \cite{limits}, in which extraction is impossible.
    
    For the backwards direction, we first delete all vertices with $v<a_1$, $a_2<v<b_1$, or $b_2<v$. Then $Y$-measuring every remaining vertex yields the desired result.
\end{proof}

\begin{figure}
\centering
\begin{tikzpicture}[scale=0.75, transform shape]
    % crossing
    \node[circle, fill=purple!60, minimum size=17pt] (1) at (-2,0) {};
    \node[circle, fill=purple!60, inner sep=2pt] (2) at (-0.5,0) {\textcolor{white}{$a_1$}};
    \node[circle, fill=purple!60, inner sep=1.6pt] (3) at (1,0) {\textcolor{white}{$b_1$}};
    \node[circle, fill=purple!60, inner sep=2pt] (4) at (2.5,0) {\textcolor{white}{$a_2$}};
    \node[circle, fill=purple!60, minimum size=17pt] (5) at (4,0) {};
    \node[circle, fill=purple!60, inner sep=1.6pt] (6) at (5.5,0) {\textcolor{white}{$b_2$}};
    
    \draw[gray, line width=1.5] (1) -- (2);
    \draw[gray, line width=1.5] (2) -- (3);
    \draw[gray, line width=1.5] (3) -- (4);
    \draw[gray, line width=1.5] (4) -- (5);
    \draw[gray, line width=1.5] (5) -- (6);

    \draw[gray!50, line width=1, <->] (1,-1.75)+(65:2.5cm) arc (65:115:2.5cm);
    \draw[gray!50, line width=1, <->] (3.25,2.8)+(240:3.5cm) arc (240:300:3.5cm);

    % within
    \node[circle, fill=purple!60, minimum size=17pt] (1) at (-2,-2.5) {};
    \node[circle, fill=purple!60, inner sep=2pt] (2) at (-0.5,-2.5) {\textcolor{white}{$a_1$}};
    \node[circle, fill=purple!60, inner sep=1.6pt] (3) at (1,-2.5) {\textcolor{white}{$b_1$}};
    \node[circle, fill=purple!60, inner sep=1.6pt] (4) at (2.5,-2.5) {\textcolor{white}{$b_2$}};
    \node[circle, fill=purple!60, minimum size=17pt] (5) at (4,-2.5) {};
    \node[circle, fill=purple!60, inner sep=2pt] (6) at (5.5,-2.5) {\textcolor{white}{$a_2$}};
    
    \draw[gray, line width=1.5] (1) -- (2);
    \draw[gray, line width=1.5] (2) -- (3);
    \draw[gray, line width=1.5] (3) -- (4);
    \draw[gray, line width=1.5] (4) -- (5);
    \draw[gray, line width=1.5] (5) -- (6);
    
    \draw[gray!50, line width=1, <->] (2.5,-6.5)+(60:5cm) arc (60:120:5cm);
    \draw[gray!50, line width=1, <->] (1.75,-1)+(255:2cm) arc (255:285:2cm);

    %adjacent
    \node[circle, fill=purple!60, minimum size=17pt] (1) at (-2,-5) {};
    \node[circle, fill=purple!60, inner sep=2pt] (2) at (-0.5,-5) {\textcolor{white}{$a_1$}};
    \node[circle, fill=purple!60, inner sep=2pt] (3) at (1,-5) {\textcolor{white}{$a_2$}};
    \node[circle, fill=purple!60, inner sep=1.6pt] (4) at (2.5,-5) {\textcolor{white}{$b_1$}};
    \node[circle, fill=purple!60, minimum size=17pt] (5) at (4,-5) {};
    \node[circle, fill=purple!60, inner sep=1.6pt] (6) at (5.5,-5) {\textcolor{white}{$b_2$}};
    
    \draw[gray, line width=1.5] (1) -- (2);
    \draw[gray, line width=1.5] (2) -- (3);
    \draw[gray, line width=1.5] (3) -- (4);
    \draw[gray, line width=1.5] (4) -- (5);
    \draw[gray, line width=1.5] (5) -- (6);
    
    \draw[gray!50, line width=1, <->] (4,-6.75)+(65:2.5cm) arc (65:115:2.5cm);
    \draw[gray!50, line width=1, <->] (0.25,-6.5)+(75:2cm) arc (75:105:2cm);

    % non-adjacent
    \node[circle, fill=purple!60, inner sep=2pt] (1) at (-2,-7.5) {\textcolor{white}{$a_1$}};
    \node[circle, fill=purple!60, inner sep=2pt] (2) at (-0.5,-7.5) {\textcolor{white}{$a_2$}};
    \node[circle, fill=purple!60, minimum size=17pt] (3) at (1,-7.5) {};
    \node[circle, fill=purple!60, inner sep=1.6pt] (4) at (2.5,-7.5) {\textcolor{white}{$b_1$}};
    \node[circle, fill=purple!60, minimum size=17pt] (5) at (4,-7.5) {};
    \node[circle, fill=purple!60, inner sep=1.6pt] (6) at (5.5,-7.5) {\textcolor{white}{$b_2$}};
    
    \draw[gray, line width=1.5] (1) -- (2);
    \draw[gray, line width=1.5] (2) -- (3);
    \draw[gray, line width=1.5] (3) -- (4);
    \draw[gray, line width=1.5] (4) -- (5);
    \draw[gray, line width=1.5] (5) -- (6);
    
    \draw[gray!50, line width=1, <->] (4,-9.25)+(65:2.5cm) arc (65:115:2.5cm);
    \draw[gray!50, line width=1, <->] (-1.25,-9)+(75:2cm) arc (75:105:2cm);

    \node[red,scale=3] at (-3.5,0) {\ding{55}};
    \node[red,scale=3] at (-3.5,-2.5) {\ding{55}};
    \node[red,scale=3] at (-3.5,-5) {\ding{55}};
    \node[green!70!blue!70,scale=3] at (-3.5,-7.5) {\ding{51}};

    \begin{scope}[scale=1.333,transform shape]
      \node at (-3.75,0) {(a)};
      \node at (-3.75,-1.875) {(b)};
      \node at (-3.75,-3.75) {(c)};
      \node at (-3.75,-5.625) {(d)};
    \end{scope}
\end{tikzpicture}
\caption{All four configurations of the Bell vertex-minor problem on line graphs. Extraction is not possible in (a) Theorem 3 in ref. \cite{limits}, (b) Proposition \ref{prop:line2}, (c) Proposition \ref{prop:line3}, and possible in (d) the remaining case.}
\end{figure}
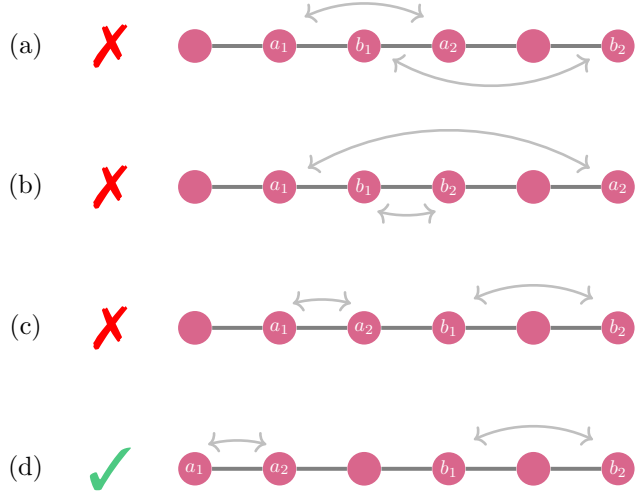

\begin{theorem}
    Two Bell pairs $\{a_1,a_2\}$ and $\{b_1,b_2\}$ can be extracted from a tree graph $G$ iff the path from $a_1$ to $a_2$ and the path from $b_1$ to $b_2$ are non-adjacent.\label{theo:treebellvm}
\end{theorem}
\begin{proof}
    We induct on the number of vertices $n$ in $G$.

    In the base case for $n=4$, $G$ is either a line graph or a star graph. In either case, extraction is impossible. The two paths must also be adjacent.

    Now assume the result holds up to some $n \geq 4$. Given some tree graph $G$ with $n+1$ vertices, first note that leaf deletion as in Theorem \ref{theo:sreduction} (and axil deletion, which is isomorphic) results in tree graphs \cite{diestel} and never changes the number of connected components. Consider the following cases:

    \begin{enumerate}[label=(\arabic*)]
        \item There is a leaf $v$ not in the target graph.
        
        By Theorem \ref{theo:sreduction}a and the inductive hypothesis, $K_2\sqcup K_2 < G \iff K_2\sqcup K_2 < (G \setminus v) \iff $ the paths on $G\setminus v$ are non-adjacent $\iff$ the paths on $G$ are non-adjacent. To show the last equivalence, consider the contrapositive. A graph with adjacent paths will still contain those adjacent paths after an addition of a leaf. Deletion of a leaf not in the target graph also doesn't change this: it can't be an endpoint of one of the paths, which are path of the target graph, it can't be an intermediate vertex, where $\deg \geq 2$, and it can't delete an edge connecting the two paths, which are between vertices in those paths.
        
        \item All leaves are in the target graph, and there is a leaf $w$ with axil $v$ not in the target graph.
        
        By Theorem \ref{theo:sreduction}b and the inductive hypothesis, $K_2\sqcup K_2 < G \iff K_2\sqcup K_2 < (\tau_w \circ \tau_v(G) \setminus v) \iff $ the paths on $\tau_w \circ \tau_v(G) \setminus v$ are non-adjacent $\iff$ the paths on $G$ are non-adjacent. Again, consider the contrapositive. If $G$ has adjacent paths, then $\tau_w \circ \tau_v(G) \setminus v$ amounts to deleting the endpoint of one of the paths, a leaf. Then the paths are still adjacent, since the edge connecting the two paths could not have been incident on that leaf. If $\tau_w \circ \tau_v(G) \setminus v$ has adjacent paths, then $G$ extends one path with a leaf at one endpoint; the paths remain adjacent.

        \item All leaves and corresponding axils are in the target graph.
        
        Consider the number of leaves in $G$. It is known in graph theory that every tree must have at least two leaves. If there are exactly two leaves, $G$ must be a line graph. Then we can apply Theorem \ref{theo:linebellvm}, which is exactly the statement of this theorem.
        
        If $G$ has three leaves, then each must have the last vertex of the target graph as its axil. Otherwise if one of the leaves has another leaf as its axil, we would have a disconnected $K_2$. But we know that $G$ is a single connected component by the definition of tree graphs. If $G$ has more than four vertices, the axil must have $\deg \geq 4$. It is known in graph theory that a tree has at least as many leaves as the vertex with maximum degree, so $G$ must have a fourth leaf not in the target graph. But that contradicts our assumption, so we can ignore this case.
        
        Finally, $G$ cannot have four or more leaves, because as above, we would have a disconnected $K_2$.
    \end{enumerate}
    So the result holds for all $n\geq 4$.
\end{proof}

The core of the proof lies in the fact that due to source reduction, we only need to consider line graphs and the four vertex star graph.

This result is noteworthy because although Bell vertex-minor in general is NP-complete, the problem for tree graphs can be solved by simple tree traversal. This is suggested by the fact that distance-hereditary graphs, which contain tree graphs, are easier to solve, as studied in ref. \cite{algo}. Admittedly, extraction is possible only in the obvious case, which works on any graph.

\subsection{Ring Graphs}

The \textbf{ring graph} is defined as $R_n(V_n, E_n)$ with $V_n = \{1,\dots, n\}$ and $E_n=\{\{1,2\}, \{2,3\}, \dots, \{n-1,n\}, \{n,1\}\}$. Ref. \cite{limits} showed that $G(\{a_1,b_1,a_2,b_2\}, \{\{a_1,a_2\},\{b_1,b_2\}\}) \nless R_n$, where $a_1<b_1<a_2<b_2$. We will also be able to solve the ring graph completely.

\begin{lemma}
    It is not possible to extract a Bell pair $\{a_1,a_2\}$ and an isolated qubit $b$ from $R_n$ if $a_1,a_2,b$ are consecutive, in that order.\label{lemma:ring1a}
\end{lemma}

\begin{proof}
    The base case for $n = 3$ is trivial. Assuming the result holds up to some $n\geq 3$, we want to show that $K_2 \sqcup \{b\} \nless P_v(R_{n+1})$ for each $P_v \in \{X_v,Y_v,Z_v\}$.
    \begin{enumerate}[label=(\arabic*)]
        \item $K_2 \sqcup \{b\} \stackrel{?}{\nless} Z_v(R_{n+1})$: The graph becomes a line graph. By repeated application of Theorem \ref{theo:sreduction}a on either end of the graph, we obtain a connected three vertex graph, where extraction is impossible.
        \item $K_2 \sqcup \{b\} \stackrel{?}{\nless} Y_v(R_{n+1})$: The graph becomes $R_n$ with $a_1,a_2,b$ still consecutive, so we can apply the inductive hypothesis.
        \item $K_2 \sqcup \{b\} \stackrel{?}{\nless} X_v(R_{n+1})$: The graph becomes $R_n$ except the two neighbors of $v$ are a leaf-axil pair. If either of this pair is not $a_1$, $a_2$, or $b$, then we can apply Theorem \ref{theo:sreduction} followed by the inductive hypothesis. If this pair is $\{a_1,b\}$ we get a connected three vertex graph, where extraction is impossible.
    \end{enumerate}
    So the result holds for all $n\geq 3$.
\end{proof}

\begin{proposition}
    It is not possible to extract two Bell pairs $\{a_1,a_2\}$ and $\{b_1,b_2\}$ from $R_n$ if $a_1,a_2,b_1$ are consecutive, in that order.\label{prop:ring2}
\end{proposition}

\begin{proof}
    The base case for $n = 4$ is trivial. WLOG let $a_1<a_2<b_1<b_2$. Assuming the result holds up to some $n\geq 4$, we want to show that $K_2 \sqcup K_2 \nless P_v(R_{n+1})$ for each $P_v \in \{X_v,Y_v,Z_v\}$.
    \begin{enumerate}[label=(\arabic*)]
        \item $K_2 \sqcup K_2 \stackrel{?}{\nless} Z_v(R_{n+1})$: If $v<a_1$ or $v>b_2$, the graph becomes the line graph in Proposition \ref{prop:line3}, where extraction is impossible. Otherwise if $b_1<v<b_2$, the graph becomes the line graph in Proposition \ref{prop:line2}, where extraction is also impossible.
        \item $K_2 \sqcup K_2 \stackrel{?}{\nless} Y_v(R_{n+1})$: Same as Lemma \ref{lemma:ring1a}.
        \item $K_2 \sqcup K_2 \stackrel{?}{\nless} X_v(R_{n+1})$: The graph becomes $R_n$ except the two neighbors of $v$ are a leaf-axil pair. If this pair is $\{a_1,b_2\}$, then $a_1 \sim_\text{F} b_2$. Lemma \ref{lemma:fgec} yields a contradiction, since none of the descriptions correspond to $\{a_1,b_2\}$ in $K_2 \sqcup K_2$. If the pair is $\{b_1,b_2\}$, then $b_1 \sim_\text{F} b_2$ in both $X_v(R_{n+1})$ and $K_2 \sqcup K_2$. So we can apply Theorem \ref{theo:treduction} with leaf deletion to get $K_2 \sqcup K_2 < X_v(R_{n+1}) \implies K_2 \sqcup \{b\} < R_n$. Deletion of the leaf in this pair reduces to the problem in Lemma \ref{lemma:ring1a}, which we know to be impossible. So $K_2 \sqcup K_2 \nless X_v(R_{n+1})$.
    \end{enumerate}
    So the result holds for all $n\geq 4$.
\end{proof}

\begin{theorem}
    Two Bell pairs $\{a_1,a_2\}$ and $\{b_1,b_2\}$ can be extracted from $R_n$ iff $a_1 < a_2 < b_1 < b_2$ (up to some reordering) and no three of $\{a_1,a_2,b_1,b_2\}$ are consecutive.\label{theo:ringbellvm}
\end{theorem}
\begin{proof}
    Consider the contrapositive of the forwards direction. All arrangements not described in the statement fall under one of the Proposition \ref{prop:ring2} or Theorem 1 in \cite{limits}, in which extraction is impossible.
    
    For the backwards direction, first consider the case where $\{a_1,b_2\}, \{a_2,b_1\} \notin E_n$. We delete all vertices between $a_1$ and $b_2$ and all vertices between $a_2$ and $b_1$. Then $Y$-measuring every remaining vertex yields the desired result.

    For the remaining case, WLOG $\{a_1,b_2\} \in E_n$ (see Figure \ref{fig:ringbellvm}). We know $\{a_1,a_2\}, \{b_1,b_2\} \notin E_n$ since we are given that no three of $\{a_1,a_2,b_1,b_2\}$ are consecutive. Then $Y$-measuring every vertex between $a_2$ and $b_1$, all but one of the vertices between $a_1$ and $a_2$, and all but one of the vertices between $b_1$ and $b_2$ results in the six vertex graph in Figure 3a in \cite{limits}, in which extraction is possible.
\end{proof}

\begin{figure}[htb]
\centering
\begin{tikzpicture}[scale=0.75, transform shape]
  %original graph
  \node[circle, fill=purple!60] (1) at (0,2) {\textcolor{white}{1}};
  \node[circle, fill=purple!60] (2) at (1.414,1.414) {\textcolor{white}{2}};
  \node[circle, fill=purple!60] (3) at (2,0) {\textcolor{white}{3}};
  \node[circle, fill=purple!60] (4) at (1.414,-1.414) {\textcolor{white}{4}};
  \node[circle, fill=purple!60] (5) at (0,-2) {\textcolor{white}{5}};
  \node[circle, fill=purple!60] (6) at (-1.414,-1.414) {\textcolor{white}{6}};
  \node[circle, fill=purple!60] (7) at (-2,0) {\textcolor{white}{7}};
  \node[circle, fill=purple!60] (8) at (-1.414,1.414) {\textcolor{white}{8}};
  
  \draw[gray, line width=1.5] (1) -- (2);
  \draw[gray, line width=1.5] (2) -- (3);
  \draw[gray, line width=1.5] (3) -- (4);
  \draw[gray, line width=1.5] (4) -- (5);
  \draw[gray, line width=1.5] (5) -- (6);
  \draw[gray, line width=1.5] (6) -- (7);
  \draw[gray, line width=1.5] (7) -- (8);
  \draw[gray, line width=1.5] (8) -- (1);

  \node[circle, fill=purple!60
  ] (1) at (7,2) {\textcolor{white}{1}};
  \node[circle, fill=purple!60] (2) at (8.414,1.414) {\textcolor{white}{2}};
  \node[circle, fill=purple!60] (3) at (9,0) {\textcolor{white}{3}};
  \node[circle, fill=purple!60] (4) at (8.414,-1.414) {\textcolor{white}{4}};
  \node[circle, fill=purple!15] (5) at (7,-2) {\textcolor{white}{5}};
  \node[circle, fill=purple!60] (6) at (5.586,-1.414) {\textcolor{white}{6}};
  \node[circle, fill=purple!60] (7) at (5,0) {\textcolor{white}{7}};
  \node[circle, fill=purple!60] (8) at (5.586,1.414) {\textcolor{white}{8}};
  
  \draw[gray, line width=1.5] (1) -- (2);
  \draw[gray, line width=1.5] (2) -- (3);
  \draw[gray, line width=1.5] (3) -- (4);
  \draw[gray, line width=1.5] (4) -- (6);
  \draw[gray, line width=1.5] (6) -- (7);
  \draw[gray, line width=1.5] (7) -- (8);
  \draw[gray, line width=1.5] (8) -- (1);

  \node[circle, fill=purple!60
  ] (1) at (7,-5) {\textcolor{white}{1}};
  \node[circle, fill=purple!60] (2) at (8.414,-5.586) {\textcolor{white}{2}};
  \node[circle, fill=purple!60] (3) at (9,-7) {\textcolor{white}{3}};
  \node[circle, fill=purple!60] (4) at (8.414,-8.414) {\textcolor{white}{4}};
  \node[circle, fill=purple!15] (5) at (7,-9) {\textcolor{white}{5}};
  \node[circle, fill=purple!60] (6) at (5.586,-8.414) {\textcolor{white}{6}};
  \node[circle, fill=purple!60] (7) at (5,-7) {\textcolor{white}{7}};
  \node[circle, fill=purple!15] (8) at (5.586,-5.586) {\textcolor{white}{8}};
  
  \draw[gray, line width=1.5] (1) -- (2);
  \draw[gray, line width=1.5] (2) -- (3);
  \draw[gray, line width=1.5] (3) -- (4);
  \draw[gray, line width=1.5] (4) -- (6);
  \draw[gray, line width=1.5] (6) -- (7);
  \draw[gray, line width=1.5] (7) -- (1);

  \node[circle, fill=purple!60
  ] (1) at (0,-5) {\textcolor{white}{1}};
  \node[circle, fill=purple!60] (2) at (1.414,-5.586) {\textcolor{white}{2}};
  \node[circle, fill=purple!15] (3) at (2,-7) {\textcolor{white}{3}};
  \node[circle, fill=purple!60] (4) at (1.414,-8.414) {\textcolor{white}{4}};
  \node[circle, fill=purple!15] (5) at (0,-9) {\textcolor{white}{5}};
  \node[circle, fill=purple!60] (6) at (-1.414,-8.414) {\textcolor{white}{6}};
  \node[circle, fill=purple!15] (7) at (-2,-7) {\textcolor{white}{7}};
  \node[circle, fill=purple!15] (8) at (-1.414,-5.586) {\textcolor{white}{8}};
  
  \draw[gray, line width=1.5] (1) -- (6);
  \draw[gray, line width=1.5] (2) -- (4);
  
  \draw[->] (3,0)--(4,0);
  \node at (3.5,0.3) {$Y_5$};
  \draw[->] (7,-3)--(7,-4);
  \node at (7.3,-3.5) {$Y_8$};
  \draw[->] (4,-7)--(3,-7);
  
  \draw[gray!50, line width=1, <->] (-1.2,-1) -- (-0.2,1.6);
  \draw[gray!50, line width=1, <->] (1.3,-1) -- (1.3,1);

\end{tikzpicture}
\caption{Process in Theorem \ref{theo:ringbellvm} illustrated for extracting Bell pairs $\{1,6\}$ and $\{2,4\}$ from an eight vertex graph. Vertices are $Y$-measured to yield the six vertex graph in Figure 3a in \cite{limits}, which delineates the steps to reach $K_2 \sqcup K_2$.}
\label{fig:ringbellvm}
\end{figure}
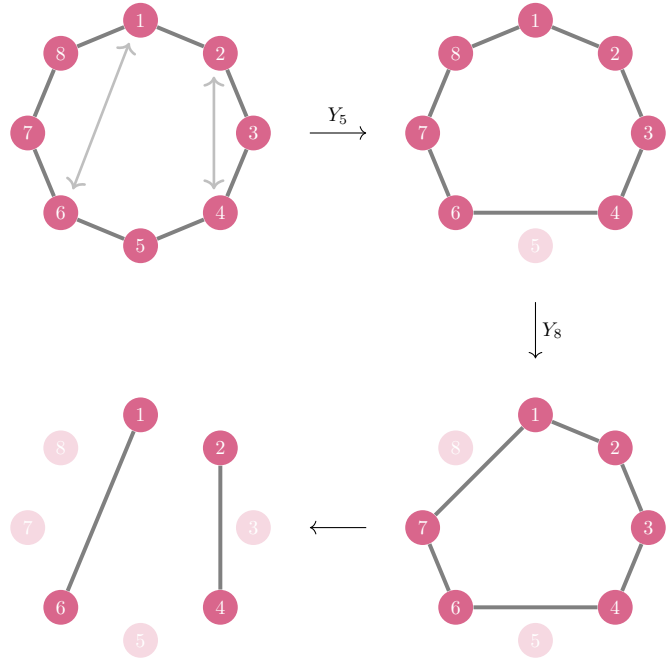

The conditions for extraction on the ring graph are identical to those for the line/tree graphs, except now extraction is possible when paths are adjacent, if both paths have more than two vertices.

\section{Discussion}

In this paper, we used recently introduced foliage concepts to completely solve the Bell vertex-minor problem for line, tree, and ring graphs. In order to fully utilize lifted local complementation, we extended the definition of the foliage partition to include a wide range of partitions on which LC-invariance is still meaningful. We used this generalization to prove a technique for reducing the vertex-minor problem, which we call ``target reduction." We also applied this definition on an existing ``source reduction" technique to illuminate the underlying foliage structure.

Our techniques could potentially be used for a variety of extraction problems and network architectures. The GHZ extraction problem was recently solved for line graphs \cite{linear}; a near identical proof to Theorem \ref{theo:treebellvm} can be used to solve the four-vertex GHZ state extraction problem on tree graphs. Similar techniques might be used to find the complete solution for tree graphs. Target reduction appears to be particularly well-suited for GHZ extraction, as every pair of vertices in the target graph are already foliage-equivalent.

We primarily applied a graph-theoretic approach to the vertex-minor problem. Problems relevant to physical quantum networks remain to be investigated. In particular, we have not attempted to minimize the number of measurements while extracting Bell pairs. Bypassing the bottlenecks we have identified using long-distance links are also necessary for physical implementations.

\section*{Acknowledgements}
\addcontentsline{toc}{section}{Acknowledgements}

This research project started at the 2023 UC Davis Math REU, which was supported by NSF grant DMS-1950928. I am incredibly grateful for Professor Bruno Nachtergaele's mentorship, without which this project could not have happened. I would also like to thank Andrew Jackson and Rahul Hingorani for their guidance, and the rest of the REU participants.

\newpage
\section*{Appendix}

\begin{proof}[Proof of Proposition \ref{prop:feer}.]
    The core of the proof is transitivity.
    \begin{itemize}
        \item \textbf{Reflexive}: Given in definition.
        \item \textbf{Symmetric}: Follows from definition.
        \item \textbf{Transitive}: Given distinct vertices $v,w,u\in G$ with $v\sim_\text{F} w$ and $w \sim_\text{F} u$, we have 6 cases:
        \begin{enumerate}[label=(\arabic*)]
            \item $\{v,w\}, \{w,u\} \in T_G$: Then $N_v \setminus w = N_w \setminus v$ and $N_w \setminus u = N_u \setminus w \implies (N_v\setminus u)\setminus w = (N_u \setminus v) \setminus w$. Consider $w \in N_v \iff v\in N_w \iff v\in N_u \iff u \in N_v \iff u \in N_w \iff w \in N_w$, which yields $N_v\setminus u = N_u \setminus v$, so $v$ and $u$ are twins.
            \item $\{v,w\} \in T_G$, $(w,u) \in L_G$: Since $N_w = \{u\}$, we must have $N_v = \{u\}$. Then $v$ is a leaf of $u$.
            \item $\{v,w\} \in T_G$, $(u,w) \in L_G$: The twin relation implies $u \in N_v$, which is impossible since $u$ has degree 1.
            \item $(v,w), (w,u) \in L_G$: Impossible.
            \item $(v,w), (u,w) \in L_G$: Then $N_v=N_u=\{w\}$ and $v$ and $u$ are twins.
            \item $(w,v), (w,u) \in L_G$: Impossible.
        \end{enumerate}
    \end{itemize}
    So $\sim_\text{F}$ is an equivalence relation.
\end{proof}

\begin{proof}[Proof of Theorem \ref{theo:fplci}.]
    Equivalently, we want to show that $v\sim_\text{F} w$ in $G$ iff $v' \sim_\text{F} w'$ in $\tau_a(G)$.
    
    If $\{v,w\} \in T_G$, consider the location of $a$:
    \begin{enumerate}[label=(\arabic*)]
        \item $a=v$ (or $a=w$) and $\{v,w\} \in E$: Then $w$ is disconnected from all its neighbors except $v$, and $(w,v)$ is now a leaf-axil pair.
        \item $a=v$ (or $a=w$) and $\{v,w\} \notin E$: Then $N_v, N_w$ are unchanged and $v$ and $w$ remain twins.
        \item $a \in N_v \setminus w$ ($\implies a \in N_w \setminus v$): Then $\tau_a$ complements the edge $\{v,w\}$ but otherwise leaves $N_v, N_w$ unchanged. So $v$ and $w$ remain twins.
        \item $a \notin N_v \cup N_w \cup \{v,w\}$: Then $N_v, N_w$ are unchanged and $v$ and $w$ remain twins.
    \end{enumerate}
    If $(v,w) \in L_G$, again consider the location of $a$:
    \begin{enumerate}[label=(\arabic*)]
        \item $a = v$: Since $N_a = w$, $\tau_a$ does nothing.
        \item $a = w$: Then $v$ gains all of the neighbors of $w$, and $v$ and $w$ are now a twin pair.
        \item $a \notin \{v,w\}$: Then $\tau_a$ cannot change the neighbors of $v$, and $v$ and $w$ remain a leaf-axil pair.
    \end{enumerate}
    Since $\tau_a$ is its own inverse, both directions are identical. Applying Theorem \ref{theo:sequence}, foliage-equivalence, and thus the foliage partition are LC-invariant.
\end{proof}

\end{document}